\documentclass[12pt]{article}

\usepackage{amsmath,amssymb,amsthm,amsfonts,mathrsfs}
\usepackage{color}
\usepackage{hyperref}
\usepackage{fullpage}
\usepackage{subfigure}
\usepackage{graphicx}
\usepackage{graphics}
\usepackage{subfigure}
\usepackage{multirow}
\usepackage{booktabs}

\newcommand{\va}[0]{\mathbf a}
\newcommand{\vb}[0]{\mathbf b}

\newcommand{\ve}[0]{\mathbf e}

\newcommand{\vv}[0]{\mathbf v}
\newcommand{\vw}[0]{\mathbf w}
\newcommand{\vx}[0]{\mathbf x}

\newcommand{\vz}[0]{\mathbf z}

\newcommand{\vzero}[0]{\mathbf 0}
\newcommand{\vone}[0]{\mathbf 1}
\newcommand{\mbf}[1]{\mbox{\boldmath$#1$}}
\newcommand{\Cov}[0]{\text{Cov}}
\newcommand{\Var}[0]{\text{Var}}

\newcommand{\diag}[0]{\text{diag}}
\newcommand{\tr}[0]{\text{tr}}

\newcommand{\E}[0]{\mathbb{E}}
\newcommand{\calX}[0]{\mathcal{X}}

\newcommand{\Prob}[0]{\mathbb{P}}
\newcommand{\I}[0]{\mathbb{I}}

\newcommand{\FP}[0]{\text{FP}}

\newtheorem{thm}{Theorem}[section]
\newtheorem{lem}[thm]{Lemma}
\newtheorem{prop}[thm]{Proposition}

\theoremstyle{definition}
\newtheorem{defn}{Definition}[section]

\newtheorem{exmp}{Example}[section]

\theoremstyle{remark}
\newtheorem{rmk}{Remark}

\begin{document}

\begin{center}
{\Large \bf Inference of high-dimensional linear models with time-varying coefficients}

\vspace{0.1in}

Xiaohui Chen, Yifeng He \\
{\it University of Illinois at Urbana-Champaign} \\ 
Email: xhchen@illinois.edu, he3@illinois.edu \\

\vspace{0.2in}
\today \\
\end{center}

\begin{abstract}
We propose a pointwise inference algorithm for high-dimensional linear models with time-varying coefficients. The method is based on a novel combination of the nonparametric kernel smoothing technique and a Lasso bias-corrected ridge regression estimator. Due to the non-stationarity feature of the model, dynamic bias-variance decomposition of the estimator is obtained. With a bias-correction procedure, the local null distribution of the estimator of the time-varying coefficient vector is characterized for iid Gaussian and heavy-tailed errors. The limiting null distribution is also established for Gaussian process errors, and we show that the asymptotic properties differ between short-range and long-range dependent errors. Here, p-values are adjusted by a Bonferroni-type correction procedure to control the familywise error rate (FWER) in the asymptotic sense at each time point. The finite sample size performance of the proposed inference algorithm is illustrated with synthetic data and an application to learn brain connectivity by using the resting-state fMRI data for Parkinson's disease.
\end{abstract}

\section{Introduction}

We consider the time-varying coefficient models (TVCM)
\begin{equation}
\label{eqn:tvcm}
y(t) = \vx(t)^\top \mbf{\beta}(t) + e(t)
\end{equation}
where $t \in [0,1]$ is the time index, $y(\cdot)$ the response process, $\vx(\cdot)$ the $p \times 1$ deterministic predictor process, $\mbf{\beta}(\cdot)$ the $p \times 1$ time varying coefficient vector, and $e(\cdot)$ the mean zero stationary error process. The response and predictors are observed at $t_i = i/n, i=1,...,n$, i.e. $y_i = y(t_i), \vx_i = \vx(t_i)$, and $e_i = e(t_i)$ with a known covariance matrix $\Sigma_e=\Cov(\ve)$ where $\ve=(e_1,\cdots,e_n)^\top$. TVCM is useful for capturing the dynamic associations in the regression models and longitudinal data analysis \cite{hooveretal1998}, and it has broad applications in biomedical engineering, environmental science, and econometrics. In this paper, we focus on the {\it pointwise} inference for the time-varying coefficient vector $\mbf\beta(t)$ in the high-dimensional double asymptotic framework $\min(p,n) \to \infty$. %Moreover, different from longitudinal setting, we consider only observations from one subject $(\vx_i, y_i)$.

Nonparametric estimation and inference of the TVCM in fixed dimension has been extensively studied, see e.g. \cite{robinson1989,hooveretal1998,clevelandgrosseshyu1991,fanzhang1999,zhangleesong2002,orbeferreirarodriguez-poo2005,cai_zw2007,zhangwu2012a,zhouwu2010}. In high dimensions, variable selection and estimation of varying-coefficient models using basis expansions have been studied in \cite{weihuangli2011,xuequ2012,ruiyizou2014}. Our primary objective is not to estimate $\mbf{\beta}(t)$, but rather to perform statistical inference on the coefficients. In particular, for any $t \in (0,1)$, we wish to test the local hypothesis, for $ j = 1, \cdots, p$,
\begin{equation}
H_{0,j,t}: \beta_j(t) = 0 \quad \text{VS} \quad H_{1,j,t} : \beta_j(t) \neq 0.
\end{equation}
By assigning p-values at each time point, we construct a sequence of estimators of the coefficient vectors that allows us to assess the uncertainty of the dynamic patterns in such as brain connectivity networks. Confidence intervals and hypothesis testing problems of lower-dimensional functionals of the high-dimensional constant coefficient vector $\mbf\beta(t) \equiv \mbf\beta, \forall t \in [0,1]$, have been studied in \cite{buhlmann2013,zhangzhang2013,javanmardmontanari2014a}. To the best of our knowledge, little has been done for inference of the high-dimensional TVCM and our goal is to fill the inference gap between the classical TVCM and the high-dimensional linear model.

While the existing literature on high-dimensional linear models is based on iid errors, \cite{buhlmann2013,zhangzhang2013,javanmardmontanari2014a}, we provide an asymptotic theory for answering the question that to which extent the statistical validity of inferences based on iid errors can hold for dependent errors. Allowing temporal dependence is of the practical interest as many datasets such as fMRI data are spatio-temporal and the errors are naturally correlated in the time domain. Theoretical analysis has revealed that the temporal dependence has delicate impact on the asymptotic rates for estimating the covariance structures, \cite{chenxuwu2013,chenxuwu2016}. Therefore, it is useful to build an inference procedure that is also robust in the time series context. The error process $e_i$ is modelled as a stationary linear process
\begin{equation}
\label{eqn:linearprocess}
e_i = \sum_{m=0}^\infty a_m \xi_{i-m},
\end{equation}
where $a_0=1$ and $\xi_i$ are iid mean-zero random variables (a.k.a. innovations) with variance $\sigma^2$. When the $\xi_i$ are normal, the linear processes of form (\ref{eqn:linearprocess}) are Gaussian processes that cover the autoregressive and moving-average (ARMA) models with iid Gaussian innovations as special cases. For the linear process, we deal with both weak and strong temporal dependences. In particular, if $a_m = O((m+1)^{-\varrho})$ and $\varrho > 1/2$, then $e_i$ is well-defined and has (i) short-range dependence (SRD) if $\varrho > 1$, (ii) long-range dependence (LRD) if $1/2 < \varrho < 1$. For the SRD processes, it is clear that $\sum_{m=0}^\infty |a_m| < \infty$ and therefore the long-run variance is finite.

The paper is organized as follows. In Section \ref{sec:mehod}, we describe our method in details. Asymptotic theory is presented in Section \ref{sec:asymptotics}. Section \ref{sec:experiments} presents some simulation results and Section \ref{sec:real_data} demonstrates an application to an fMRI dataset. The paper concludes in Section \ref{sec:discussion} with a discussion of some future work. Proofs and some implementation issues are available in the Appendix.

\section{Method}
\label{sec:mehod}

\subsection{Notations and Preliminary}

Let $K$ be a non-negative symmetric function with bounded support in $[-1,1], \int_{-1}^1 K(x)dx = 1$, and let $b_n$ be a bandwidth parameter satisfying $b_n = o(1)$ and $n^{-1} = o(b_n)$. For each time point $t \in \varpi = [b_n, 1-b_n]$, the Nadaraya-Waston smoothing weight is defined as
\begin{equation}
\label{eqn:nw-weights}
w(i,t) = \begin{cases}\frac{K_{b_n}( t_i-t )}{\sum_{m=1}^n K_{b_n}( t_m-t )} & \mbox{if } \vert t_i-t \vert \leq b_n \\ \makebox[70pt][c]{0} & \mbox{otherwise}\end{cases},
\end{equation}
where $K_b(\cdot) = K(\cdot/b)$. Let $N_t = \{i: |t_i-t| \leq b_n\}$ be the $b_n$-neighborhood of time $t$, $|N_t|$ be the cardinality of the discrete set $N_t$, $W_t = \text{diag}(w(i,t)_{i \in N_t})$ be the $|N_t| \times |N_t|$ diagonal matrix with $w(i,t), i \in N_t$ on the diagonal, and let $\mathcal{R}_t = \text{span}(\vx_i : i \in N_t)$ be the subspace in $\mathbb{R}^p$ spanned by $\vx_i$, the rows of design matrix $X$ in the $N_t$ neighborhood. Let $\mathcal{X}_t = (w(i,t)^{1/2} \vx_i)_{i \in N_t}^\top$, ${\cal Y}_t = (w(i,t)^{1/2} y_i)_{i \in N_t}^\top$, and ${\cal E}_t = (w(i,t)^{1/2} e_i)_{i \in N_t}^\top$. Denote $I_p$ as the $p \times p$ identity matrix. We write the singular value decomposition (SVD) of  $\mathcal{X}_t$ as
\begin{equation}
\label{eqn:svd_X}
\mathcal{X}_t = P D Q^\top
\end{equation}
where $P$ and $Q$ are $|N_t| \times r$, and $p \times r$ matrices such that $P^\top P = Q^\top Q = I_r$, and $D = \diag(d_1, \cdots, d_r)$ is a diagonal matrix containing the $r$ nonzero singular values of ${\cal X}_t$. Now let $P_{\mathcal{R}_t}$ be the projection matrix onto $\mathcal{R}_t$,
\begin{equation}
P_{\mathcal{R}_t} = \mathcal{X}_t^\top (\mathcal{X}_t \mathcal{X}_t^\top)^- \mathcal{X}_t = Q Q^\top,
\end{equation}
where $(\mathcal{X}_t \mathcal{X}_t^\top)^- = P D^{-2} P^\top$ is the pseudo-inverse matrix of $\mathcal{X}_t \mathcal{X}_t^\top$. Let $\mbf{\theta}(t) = P_{\mathcal{R}_t}\mbf{\beta}(t)$ be the projection of $\mbf{\beta}(t)$ onto $\mathcal{R}_t$ such that $B(t) = \mbf{\theta}(t) - \mbf{\beta}(t)$ is the projection bias. Let
\begin{equation}
\label{eqn:Omega}
\Omega(\lambda) = ({\cal X}_t^\top {\cal X}_t + \lambda I_p)^{-1} {\cal X}_t^\top W_t^{1/2} \Sigma_{e,t} W_t^{1/2} {\cal X}_t ({\cal X}_t^\top {\cal X}_t + \lambda I_p)^{-1}
\end{equation}
be the covariance matrix of the time-varying ridge estimator defined in (\ref{eqn:tv-ridge}), where $\Sigma_{e,t} = \Cov((e_i)_{i \in N_t})$ and $\lambda > 0$ is the shrinkage parameter of the ridge estimator. Let $\Omega_{\min}(\lambda) = \min_{j \le p} \Omega_{jj}(\lambda)$ be the smallest diagonal entry of $\Omega(\lambda)$. For a generic vector $\vb \in \mathbb{R}^p$, we write $|\vb|_q = (\sum_{j=1}^p |b_j|^q)^{1/q}$ if $q > 0$, and $|\vb|_0=\sum_{j=1}^p \vone(b_j\neq0)$ if $q=0$. Let $\underline{w}_t = \inf_{i \in N_t} w(i,t)$ and $\overline{w}_t = \sup_{i \in N_t} w(i,t)$. For an $n \times n$ square symmetric matrix $M$ and an $n \times m$ rectangle matrix $R$, we use $\rho_i(M)$ and $\sigma_i(R)$ to denote the $i$-th largest eigenvalues of $M$ and singular values of $R$, respectively. If $k = \text{rank}(R)$, then $\sigma_1(R) \ge \sigma_2(R) \ge \cdots \ge \sigma_k(R) > 0 = \sigma_{k+1}(R) = \cdots = \sigma_{\max(m, n)}(R)$, zeros being padded to the last $\max(m, n)-k$ singular values. We take $\rho_{\max}(M)$, $\rho_{\min}(M)$ and $\rho_{\min \neq 0}(M)$ as the maximum, minimum and nonzero minimum eigenvalues of $M$, respectively, and $|M|_{\infty} = \max_{1 \le j, k \le p} |M_{jk}|$. Let 
$$
\rho_{\max}(M,s)=\max_{|\vb|_0 \le s, \vb \neq \vzero} {\vb^\top M \vb \over \vb^\top \vb}.
$$ If $M$ is nonnegative definite, then $\rho_{\max}(M,s)$ is the restricted maximum eigenvalues of $M$ at most $s$ columns and rows.

The $p$-dimensional coefficient vector $\mbf\beta(t)$ is decomposed into two parts via projecting onto the $|N_t|$-dimensional linear subspace spanned by the rows of ${\cal X}_t$ and its orthogonal complement; see Figure \ref{fig:bias-correction}. A key advantage of this decomposition is that the projected part can be conveniently estimated in closed-form, for example, by the ridge estimator since it lies in the row space of ${\cal X}_t$ and thus is amenable for the subsequent inferential analysis. In the high-dimensional situation, this projection introduces a non-negligible shrinkage bias in estimating $\mbf\beta(t)$ and therefore we may lose information because $p \gg |N_t|$. On the other hand, the shrinkage bias can be corrected by a consistent estimator of $\mbf\beta(t)$. As a particular example, we use the Lasso estimator, though any sparsity-promoting estimator attaining the same convergence rate as the Lasso should work. Because of the time-varying nature of the nonzero functional $\mbf\beta(t)$, the smoothness on the row space of ${\cal X}_t$ along the time index $t$ is necessary to apply nonparametric smoothing technique; see Fig. \ref{fig:nonstationary-rowspace}. As a special case, when the nonzero components $\mbf\beta(t)\equiv \mbf\beta$ are constant functions and the error process is iid Gaussian, our algorithm is the same as that of \cite{buhlmann2013}. Here, we emphasize that (i) coefficient vectors are time-varying (i.e. non-constant), (ii) errors are allowed to have heavy-tails by assuming milder polynomial moment conditions and to have temporal dependence, including both SRD and LRD processes. There are other inferential methods for high-dimensional linear models such as \cite{zhangzhang2013,javanmardmontanari2014a}. We do not explore specific choices here since our contribution is a general framework of combining nonparametric smoothing and bias-correction methods to make inference for high-dimensional TVCM. However, we expect that a non-stationary generalization would be feasible for those methods as well. Some simulation comparisons are provided for time-varying versions of the bias-correction methods in Section \ref{sec:experiments}.

\begin{figure}[t!]
  \begin{center}
    \subfigure[Bias correction by projection to the row space of ${\cal X}_t$.]{\label{fig:bias-correction}\includegraphics[scale=0.45]{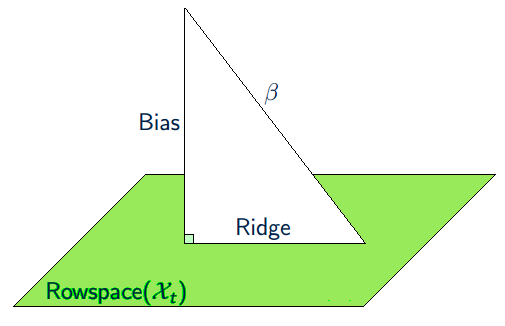}}
    \subfigure[Smoothly time-varying row space of ${\cal X}_t$.]{\label{fig:nonstationary-rowspace}\includegraphics[scale=0.5]{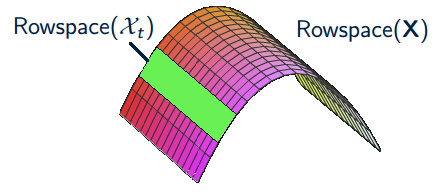}} \\
  \end{center}
  \caption{Intuition of the proposed algorithm in Section \ref{subsec:infer_alg}.}
  \label{fig:graphical_intuition}
\end{figure}

\subsection{Inference algorithm}
\label{subsec:infer_alg}

First, we estimate the projection bias $B(t)$ by $\tilde{B}(t) = (P_{\mathcal{R}_t}- I_p) \tilde{\mbf{\beta}}(t)$, where $ \tilde{\mbf{\beta}}(t) $ is the time-varying Lasso (tv-Lasso) estimator
\begin{eqnarray}
\label{eqn:tv-lasso}
\tilde{\mbf{\beta}}(t) &=& \arg\min_{\vb \in \mathbb{R}^p} \sum_{i \in N_t} w(i,t)(y_i - \vx_i^\top \vb)^2 + \lambda_1 |\vb|_1 \\
\nonumber
&=& \arg\min_{\vb \in \mathbb{R}^p} |{\cal Y}_t-{\cal X}_t \vb|_2^2 + \lambda_1 |\vb|_1.
\end{eqnarray}
Next, we estimate $\mbf{\theta}(t) = P_{\mathcal{R}_t}\mbf{\beta}(t)$ using the time-varying ridge (tv-ridge) estimator
\begin{align}
\nonumber
\tilde{\mbf\theta}(t) &= \arg\min_{\vb\in \mathbb{R}^p} \sum_{i \in N_t} w(i,t)(y_i - \vx_i^\top \vb)^2 + \lambda_2 |\vb|_2^2 \\
\label{eqn:tv-ridge}
&= (\mathcal{X}_t^\top \mathcal{X}_t + \lambda_2 I_p)^{-1} \mathcal{X}_t^\top {\cal Y}_t.
\end{align}
We defer the discussion of tuning parameters choice $\lambda_1$ and $\lambda_2$ to Section \ref{sec:asymptotics}. Our tv-Lasso bias-corrected tv-ridge regression estimator for $\mbf\beta(t)$ is
\begin{equation}
\label{eqn:proposed-estimator}
\hat{\mbf\beta}(t) = \tilde{\mbf\theta}(t) - \tilde{B}(t).
\end{equation}
Based on $\hat{\mbf\beta}(t) = (\hat\beta_1(t), \cdots, \hat\beta_p(t))^\top$, we calculate the raw two-sided p-values for individual coefficients
\begin{equation}
\label{eqn:raw-p-values}
\tilde{P}_j = 2 \left[ 1-\Phi \left( {|\hat\beta_j(t)| - \lambda_1^{1-\xi} \max_{k \neq j} |(P_{{\cal R}_t})_{jk}| \over \Omega_{jj}^{1/2}(\lambda_2) } \right) \right], \qquad j = 1, \cdots, p,
\end{equation}
where $\xi \in [0,1)$ is user pre-specified number that depends on the number of nonzero $\mbf\beta(t)$. In particular, if $|\text{supp}(\mbf\beta(t))|$ is bounded, then we can choose $\xi = 0$. Generally, following \cite{buhlmann2013}, we use $\xi = 0.05$ in our numeric examples to allow the number of nonzero components in $\mbf\beta(t)$ to diverge at proper rates. Let $\vv(t) = (V_1(t),\cdots,V_p(t))^\top \sim N(\vzero, \Omega(\lambda_2))$ and define the distribution function
\begin{equation}
\label{eqn:bonferroni_correction}
F(z) = \Prob\left( \min_{j \le p}  2 \left[ 1-\Phi \left( \Omega_{jj}^{-1/2} (\lambda_2) |V_j(t)| \right) \right] \le z \right).
\end{equation}
We adjust the $\tilde{P}_j$ for multiplicity by $P_j = F(\tilde{P}_j + \zeta)$, where $\zeta$ is another pre-defined small number \cite{buhlmann2013} that accommodates asymptotic approximation errors. Our decision rule is defined as: reject $H_{0,j,t}$ if $P_j \le \alpha$ for $\alpha \in (0,1)$. For iid errors, since $\Sigma_e = \sigma^2 I_n$ and
$$
\Omega(\lambda_2) = \sigma^2 ({\cal X}_t^\top {\cal X}_t + \lambda_2 I_p)^{-1} {\cal X}_t^\top W_t {\cal X}_t ({\cal X}_t^\top {\cal X}_t + \lambda_2 I_p)^{-1},
$$
we see that $F(\cdot)$ is independent of $\sigma$. Therefore, $F(\cdot)$ can be easily estimated by repeatedly sampling from the multivariate Gaussian distribution $N(\vzero, \Omega(\lambda_2))$. A similar observation has been made in \cite{buhlmann2013}.

\section{Asymptotic results}
\label{sec:asymptotics}

In this section, we present the asymptotic theory of the inference algorithm in Section \ref{subsec:infer_alg}. First, we state the main assumptions for iid Gaussian errors.
\begin{enumerate}
\item {\bf Error.} The errors $e_i \sim N(0, \sigma^2)$ are independent and identically distributed (iid).

\item {\bf Sparsity.} $\mbf\beta(\cdot)$ is uniformly $s$-sparse, i.e. $\sup_{t \in [0,1]} |S^*_t| \le s$, where $S^*_t = \{j : \beta_j(t) \neq 0\}$ is the support set.

\item {\bf Smoothness.}
\begin{enumerate}
\item $\mbf\beta(\cdot)$ is twice differentiable with bounded and continuous first and second derivatives in the coordinatewise sense, i.e. $\beta_j(\cdot) \in {\cal C}^2([0,1], C_0)$ for each $j = 1, \cdots, p$ and $C_0$ is an upper bound for the partial derivatives.
\item The $b_n$-neighborhood covariance matrix $\hat\Sigma_t^\diamond = |N_t|^{-1} \sum_{i \in N_t} \vx_i \vx_i^\top := {{\cal X}_t^\diamond}^\top {\cal X}_t^\diamond$ satisfies
\begin{equation}
\label{eqn:max_eigenvalue_Gram}
\rho_{\max}(\hat\Sigma_t^\diamond,s) \le \varepsilon_0^{-2} < \infty.
\end{equation}
\end{enumerate}

\item {\bf Non-degeneracy.}
\begin{eqnarray}
\label{eqn:min_diag_invGram}
\liminf_{\lambda \downarrow 0} \Omega_{\min}(\lambda) > 0.
\end{eqnarray}

\item {\bf Identifiability.} 
\begin{enumerate}
\item The minimum nonzero eigenvalue condition
\begin{equation}
\label{eqn:restricted_min_eigenvalue_Gram}
\rho_{\min \neq 0}(\hat\Sigma_t^\diamond) \ge \varepsilon_0^2 > 0.
\end{equation}

\item The {\it restricted eigenvalue condition}
\begin{equation}
\label{eqn:comptability_condtion}
\phi_0 = \inf \left\{ \phi > 0 : \min_{|S| = s} \;\; \inf_{|\vb_{S^c}|_1 \le 3 |\vb_S|_1} {\vb^\top \hat\Sigma_t \vb \over |\vb_{S}|_1^2} \ge {\phi^2 \over s } \;\;   \text{   holds for all   }  t \in [0,1] \right\} > 0,
\end{equation}
where $\hat\Sigma_t = {\cal X}_t^\top {\cal X}_t$ is the kernel smoothed covariance matrix of the predictors.
\end{enumerate}

\item {\bf Kernel.} The kernel function $K(\cdot)$ is nonnegative, symmetric around 0 with bounded support in $[-1,1]$.
\end{enumerate}
Here, we comment the assumptions and their implications. Assumption 1 and 6 are standard. The Gaussian distribution is non-essential and can be relaxed to sub-Gaussian and heavier tailed distributions (Theorem \ref{thm:non-gaussian}). Assumption 2 is a sparsity condition for the nonzero functional components and allows that $s \to \infty$ slower than $\min(p,n)$. It is a key condition for maintaining the low-dimensional structure when the dimension $p$ grows with the sample size $n$. By the argument of Theorem 5 in \cite{zhoulaffertywasserman2010}, it implies that the number of the first and second non-vanishing derivatives of $\mbf\beta(t)$ is bounded by $s$ almost surely on $[0,1]$. Assumption 3 ensures the smoothness of the time-varying coefficient vectors and the design matrix so that nonparametric smoothing techniques are applicable. Examples of Assumption 3(a) include the quadratic functions $\mbf\beta(t) = \mbf\beta + \mbf\alpha t + \mbf\xi t^2 / 2$ and the periodic functions $\mbf\beta(t) = \mbf\beta + \mbf\alpha \sin(t) + \mbf\xi \cos(t)$ with $|\mbf\alpha|_\infty + |\mbf\xi|_\infty \le C_0$. Assamption 3(b) can be viewed as Lipschitz continuity on the local design matrix that is smoothly evolving, \cite{zhouwu2010}. It is weaker than the condition that $\rho_{\max}(\hat\Sigma_t^\diamond) \le \varepsilon_0^{-2}$ because the latter may grow to infinity much faster than the restricted form (\ref{eqn:max_eigenvalue_Gram}). Assumption 4 is required for a non-degenerated stochastic component of the proposed estimator which is used for the inference purpose. Assumption 5(a) and 5(b) together impose the identifiability conditions for recovering the coefficient vectors. The analogous condition of the time-invariant version has been extensively used in literature to derive theoretical properties of the Lasso model; see e.g. \cite{bickelritovtsybakov2009,vandegeerbuhlmann2009}.

For the tv-lasso bias-corrected tv-ridge estimator (\ref{eqn:proposed-estimator}), we have the following representation theorem.

\begin{thm}[Representation]
\label{thm:representation}
Fix $t \in \varpi$ and let
\begin{equation}
\label{eqn:lambda_0+L_t}
L_{t,\ell} = \max_{j \le p} \left[ \sum_{i \in N_t} w(i,t)^\ell X_{ij}^2 \right]^{1/2}, \; \ell=1,2,\cdots, \qquad \lambda_0 = 4 \sigma L_{t,2} \sqrt{\log{p}}, 
\end{equation}
and $\lambda_1 \ge 2(\lambda_0 + 2 C_0 L_{t,1} b_n (s|N_t| \overline{w}_t)^{1/2} \varepsilon_0^{-1})$. If $\lambda_2=o(1)$, Assumptions 1-6 hold, and $C \le |N_t| \underline{w}_t \le |N_t| \overline{w}_t \le C^{-1}$ for some $C \in (0,1)$, then $\hat{\mbf\beta}(t)$ admits the decomposition
\begin{eqnarray}
\label{eqn:decomposition}
\hat{\mbf\beta}(t) &=& \mbf\beta(t) + \vz(t) + \mbf\gamma(t), \\
\label{eqn:decomposition_stochastic}
\vz(t) &\sim& N(\vzero, \Omega(\lambda_2)), \\
\label{eqn:decomposition_bias}
|\gamma_j(t)| &\le&  {\lambda_2 |\mbf\theta(t)|_2 + 2 C_0 s^{1/2} b_n \over C \varepsilon_0^2} + {4 \lambda_1 s \over \phi_0^2} |P_{{\cal R}_t} - I_p|_{\infty}, \quad j = 1,\cdots,p,
\end{eqnarray}
with probability tending to one. If $\beta_j(t) = 0$, then we have
\begin{equation}
\label{eqn:decomposition_marginal}
\Omega^{-1/2}_{jj}(\lambda_2) ( \hat\beta_j(t) - \gamma_j(t)) \sim N(0,1),
\end{equation}
where
\begin{equation}
\label{eqn:gamma_marginal_beta_j=0}
|\gamma_j(t)| \le {\lambda_2 |\mbf\theta(t)|_2 + 2 C_0 s^{1/2} b_n \over C \varepsilon_0^2} + {4 \lambda_1 s \over \phi_0^2} \max_{k \neq j} |(P_{{\cal R}_t})_{jk}|
\end{equation}
with probability tending to one.
\end{thm}

\begin{rmk}
The decomposition (\ref{eqn:decomposition}) can be viewed as a {\it local version} of the one proposed in \cite{buhlmann2013} (Proposition 2). However, due to the time-varying nature of the nonzero coefficient vectors, both the stochastic component $\vz(t)$ in (\ref{eqn:decomposition_stochastic}) and the bias component $\mbf\gamma(t)$ in (\ref{eqn:decomposition_bias}) differ from \cite{buhlmann2013}. First, our bound (\ref{eqn:decomposition_bias}) for bias has three terms arising from: ridge shrinkage, {\it non-stationarity} and Lasso correction, and each has localized features depending on the bandwidth $b_n$ of the sliding window and the smoothness parameter $C_0$. Second, the stochastic part (\ref{eqn:decomposition_stochastic}) also has time-dependent features in the covariance matrix (i.e. $\Omega(\lambda_2)$ implicitly depends on $t$ though ${\cal X}_t$) and the scale of normal random vector is different from \cite{buhlmann2013}. Delicate balance among them allows us to perform valid statistical inference such as hypothesis testing and confidence interval construction for the coefficients and, more broadly, their lower-dimensional linear functionals.
\end{rmk}

\begin{exmp}
\label{exmp:uniform_kernel}
Consider the uniform kernel $K(x) = 0.5 \I(|x| \le 1)$ as an important special case, the kernel used for our numeric experiments in Section \ref{sec:experiments}. In this case, $\underline{w}_t = (2 n b_n)^{-1}$ and $|N_t| \underline{w}_t = |N_t| \overline{w}_t =1$. It is easily verified that under the local null hypothesis $H_{0,j,t}$, (\ref{eqn:gamma_marginal_beta_j=0}) can be simplified to 
\begin{equation*}
\gamma_j(t) = O \left( \lambda_2 |\mbf\theta(t)|_2 + s^{1/2} b_n + \lambda_1 s \max_{k \neq j} |(P_{{\cal R}_t})_{jk}| \right).
\end{equation*}
From this, it is clear that the three terms correspond to bias of ridge-shrinkage, non-stationarity and Lasso-correction. The first and last components have dynamic features and the non-stationary bias is controlled by the bandwidth and sparsity parameters. The condition $C \le |N_t| \underline{w}_t \le |N_t| \overline{w}_t \le C^{-1}$ in Theorem \ref{thm:representation} rules out the case that the kernel does not use the boundary rows in the localized window and therefore avoids any jump in the time-dependent row subspaces.
\end{exmp}

\begin{rmk}
In Theorem \ref{thm:representation}, the penalty level for the tv-Lasso estimator $\lambda_1$ can be chosen as $O(\sigma L_{t,2} \sqrt{\log{p}} + L_{t,1} s^{1/2} b_n)$. The second term in the penalty is due to the non-stationarity of $\mbf\beta(t)$ and the factor $s^{1/2}$ arises from the weak coordinatewise smoothness requirement on its derivatives (Assumption 3(a)). In the Lasso case with $\mbf\beta(t)\equiv\mbf\beta$ and $w(i,t)\equiv n^{-1}$, an ideal order of the penalty level $\lambda_1$ is $\sigma n^{-1} \max_{j\le p} (\sum_{i=1}^n X_{ij}^2)^{1/2} (\log{p})^{1/2}$ see e.g. \cite{bickelritovtsybakov2009}. In the standardized design case $n^{-1} \sum_{i=1}^n X_{ij}^2 = 1$ so that $L_{t,1}=1$ and $L_{t,2}=n^{-1/2}$, the Lasso penalty is $O(\sigma (n^{-1} \log{p})^{1/2})$, while the tv-Lasso has an additional term $s^{1/2} b_n$ that may cause a larger bias. In our case, we estimate the time-varying coefficient vectors by smoothing the data points in the localized window. Thus, it is unnatural to standardize the reweighted local design matrix to have unit $\ell^2$ length and the additional bias $O(s^{1/2}b_n)$ is due to non-stationarity. If the $X_{ij}$ are iid Gaussian random variables without standardization and we interpret the linear model as conditional on $X$, then, under the uniform kernel, we have $L_{t,2}^2 = O_\Prob(\log{p}/|N_t|)$ and, in the Lasso case, the penalty level is $O(\sigma |N_t|^{-1/2} \log{p})$. If $s = O(\log{p})$ and $b_n = O((\log{p}/n)^{1/3})$, then the choice in Theorem \ref{thm:representation} has the same order as the Lasso with constant coefficient vector.
\end{rmk}

Based on Theorem \ref{thm:representation}, we can prove that the inference algorithm in Section \ref{subsec:infer_alg} asymptotically controls the familywise error rate (FWER). Let $\alpha \in (0,1)$ and $\FP_\alpha(t)$ be the number of false rejections of $H_{0,j,t}$ based on the adjusted p-values. In the asymptotic statement, $p := p(n)$ is a function of $n$ such that $p \to \infty$ as $n \to \infty$.

\begin{thm}[Pointwise inference: multiple testing]
\label{thm:multiple_testing}
If the conditions of Theorem \ref{thm:representation} hold and 
\begin{equation}
\label{eqn:var_kills_bias_ridge}
\lambda_2 |\mbf\theta(t)|_2 + s ^{1/2} b_n = o(\Omega_{\min}(\lambda_2)^{1/2}),
\end{equation}
then we have for each fixed $t \in \varpi$
\begin{equation}
\limsup_{n \to \infty} \Prob(\FP_\alpha(t) > 0) \le \alpha.
\end{equation}
\end{thm}

The proof of Theorem \ref{thm:multiple_testing} is standard by combining the argument of Theorem 2 in \cite{buhlmann2013} and Theorem \ref{thm:representation}. Therefore, we omit the proof. Condition (\ref{eqn:var_kills_bias_ridge}) essentially requires that the shrinkage and non-stationarity biases of the tv-ridge estimator together are dominated by the variance; see also the representation (\ref{eqn:decomposition}), (\ref{eqn:decomposition_stochastic}), (\ref{eqn:decomposition_bias}), and (\ref{eqn:decomposition_marginal}). This is mild condition for two reasons. First, in view that variance of the tv-ridge estimator is lower bounded when $\lambda_2$ is small enough; c.f. (\ref{eqn:min_diag_invGram}), the first term is quite weak in the sense that the tv-ridge estimator acts on a much smaller subspace with dimension $|N_t|$ than the original $p$-dimensional vector space. Second, for the choice of penalty parameter of $\lambda_1$ in Theorem \ref{thm:representation}, the term $s^{1/2}b_n$ in (\ref{eqn:var_kills_bias_ridge}) is at most $\lambda_1$. Hence, the bias correction (including the projection and non-stationary parts) in the inference algorithm (\ref{eqn:raw-p-values}) has a dominating effect on the second term of (\ref{eqn:var_kills_bias_ridge}). Consequently, provided $\lambda_2$ is small enough, the bias correction step in computing the raw p-value asymptotically approximates the stochastic component in the tv-ridge estimator.

\begin{rmk}
\label{rmk:fdr}
The Bonferroni correction (\ref{eqn:bonferroni_correction}) for the raw p-values is often conservative and thus it may be sub-optimal in power. In our simulation studies, it seems that detection power is reasonable while the FWER is controlled at $0.05$; c.f. Table \ref{tab:simulation_n=300_p=300_s=3_1} and \ref{tab:simulation_n=300_p=300_s=3_1_ctd}. To improve the power, one can consider the control of the false discovery proportion (FDP) by the principal factor approximation (PFA) method proposed in \cite{fanhangu2012,fanha2016}. By Theorem \ref{thm:representation}, under the global null hypothesis $H_{0,t}: \beta_1(t) = \cdots = \beta_p(t) = 0$, we have $\hat{\mbf\beta}(t) - \mbf\gamma(t) \sim N(\vzero, \Omega(\lambda_2))$ with a known covariance matrix $\Omega(\lambda_2)$. Therefore, our test statistic is jointly normal and the PFA can be applied to control the FDP if $\Omega(\lambda_2)$ can be well approximated by the covariance matrix of a factor model plus a weakly dependent component. \cite{fanhangu2012} provided a practical procedure to estimate the FDP.
\end{rmk}

Next, we relax the iid assumption on the errors to allow temporal dependence.

\begin{thm}[Gaussian process errors]
\label{thm:representation_GP}
Suppose that the error process $e_i$ is a mean-zero stationary Gaussian process of form (\ref{eqn:linearprocess}) such that $|a_m| \le K (m+1)^{-\varrho}$ for some $\varrho \in (1/2,1) \cup (1,\infty)$ and finite constant $K> 0$. Under Assumptions 2-6 and the notation of Theorem \ref{thm:representation} with
\begin{equation}
\label{eqn:lambda_0+L_t_GP}
\lambda_0 = \left\{ 
\begin{array}{cc}
4 \sigma L_{t,2} |\va|_1 \sqrt{\log{p}} \quad & \text{if } \varrho > 1 \\
C_{\varrho,K} \sigma L_{t,2} n^{1-\varrho} \sqrt{\log{p}} \quad & \text{if } 1 > \varrho > 1/2 \\
\end{array} \right. ,
\end{equation}
where $\va = (a_0, a_1,\cdots)^\top$, we have the representation of $\hat{\mbf\beta}(t)$ in (\ref{eqn:decomposition})--(\ref{eqn:gamma_marginal_beta_j=0}) with probability tending to one.
\end{thm}

From Theorem \ref{thm:representation_GP}, the temporal dependence strength has a dichotomous effect on the choice of $\lambda_0$, and therefore on the asymptotic properties of $\hat{\mbf\beta}(t)$. For $e_i$ with SRD, we have $|\va|_1 < \infty$ and $\lambda_0 \asymp \sigma L_{t,2} \sqrt{\log{p}}$. Therefore, the bias-correction part $\mbf\gamma(t)$ of estimating $\mbf\beta(t)$ has the same rate of convergence as the iid error case. The temporal effect only plays a role in the long-run covariance matrix of the stochastic part $\vz(t)$. If $e_i$ has LRD, then the temporal dependence has impact on both $\mbf\gamma(t)$ and $\vz(t)$. In addition, the choice of the bandwidth parameter $b_n$ is different from the SRD and iid cases. In particular, the optimal bandwidth for $\varrho \in (1/2,1)$ is $O((\log{p}/n^\varrho)^{1/3})$ which is much larger than $O((\log{p}/n)^{1/3})$ in the iid and SRD cases, assuming $s$ is bounded. The boundary case $\varrho=1$ can also be characterized; details are omitted.

We also relax the moment condition on the errors that, in the iid error case, are assumed to be zero-mean Gaussian. First, it is easy to relax this assumption to distributions with sub-Gaussian tails (see Definition \ref{defn:subgaussian_rv} in the Supplementary Material) and Theorem \ref{thm:representation} and \ref{thm:multiple_testing} continue to hold, in view that the large deviation inequality and the Gaussian approximation for a weighted partial sum of the error process only depend on the tail behavior and therefore on moments of $e_i$. Second and more importantly, the sub-Gaussian assumption may be knocked down to allow iid noise processes with algebraic tails, or equivalently $e_i$ with moments up to a finite order. The consequence of this relaxation is that a larger penalty parameter for the tv-Lasso is needed for errors with polynomial moments. Let $\Xi$ be the square root matrix of $\Omega(\lambda_2) / \sigma^2$ (i.e. $\Omega(\lambda_2) = \sigma^2 \Xi \Xi^\top$) and $\mbf\xi_j$ be the $j$-th row of $\Xi$.

\begin{thm}[Heavy-tailed errors]
\label{thm:non-gaussian}
Under the conditions of Theorem \ref{thm:representation} with $\E|e_i|^q < \infty, q > 2$, choose
\begin{equation}
\lambda_0 = C_q \max \left\{ (p \mu_{n,q})^{1/q}, \; \sigma L_{t,2} (\log{p})^{1/2} \right\}, \qquad \text{for large enough  } C_q > 0,
\end{equation}
where $\mu_{n,q} = \sum_{i \in N_t} |w(i,t) X_{ij}|^q$. If $|\mbf\xi_j|_q = o(|\mbf\xi_j|_2)$ for all $j=1,\cdots,p$, then (\ref{eqn:decomposition}) holds with probability tending to one and Theorem \ref{thm:multiple_testing} holds.
\end{thm}
The assumption $|\mbf\xi_j|_q = o(|\mbf\xi_j|_2)$ is needed to ensure the asymptotic validity of the Gaussian approximation of the ridge component (\ref{eqn:decomposition_stochastic}) for non-Gaussian data.

\section{Simulation studies}
\label{sec:experiments}

\subsection{Simulation setup}

In the simulation studies, we generated the $n \times p$ design matrix with iid rows sampled from $N(\vzero,\Sigma_\mathbf{X})$ for $n = 300$ and $p = 300$. We considered two covariance structures on the design matrix: (i) $\Sigma_\mathbf{X} = I_p$; (ii) $\Sigma_{\mathbf{X}} = T$, where $T = (t_{jk})_{j,k=1}^p$ and $t_{jk} = 0.5^{\vert j-k \vert}$. The time-varying coefficient vectors $\mbf\beta(t)$ had $s = 3$ non-zero elements and $p - 3$ zeros for all $t \in [0,1]$. The non-zero elements in $\mbf\beta(t)$ were generated by sampling nodes from a uniform distribution $U(-2.5, 2.5)$ at regular time points and smoothly interpolating on the interval $[0,1]$ using the cubic splines.  We simulated the following stationary error processes.
\begin{enumerate}
\item The $e_i$ are iid $N(0, 1)$.
\item $e_i = \varphi e_{i-1} + \xi_i$ is an AR(1) process where $\varphi \in \{0.2, 0.5\}$ and $\xi_i$ are iid $N(0,1)$.
\item The $e_i$ are iid Student's $t(3) / \sqrt{3}$.
\item $e_i$ is a long-memory process $e_i = \sum_{m=0}^\infty (m+1)^{-\varrho} \xi_{i-m}$, where $\varrho=0.75$ and $\xi_i$ are iid Gaussian with mean zero.
\end{enumerate}

We compared the performance of the proposed method with the following.
\begin{enumerate}
\item (TV-Lasso) - The time-varying Lasso, the kernel smoothed time-varying LASSO defined in (\ref{eqn:tv-lasso}), where $\lambda_1$ is selected by the cross-validation (CV). 
\item (FP-Lasso) - The false-positive Lasso, where $\lambda_1$ is tuned to match the FWER of the proposed method. This allowed us to compare the power at similar levels of FWER. 
\item (TV-LDPE) - An adaptation of the de-biased LASSO inference procedure by \cite{zhangzhang2013} to the kernel smoothed, time-varying setting.
\item (TV-SDL) - An adaptation of the SDL test of \cite{javanmardmontanari2014a} to the kernel smoothed, time-varying setting. 
\item (Non-TV) - The original non-time-varying method of \cite{buhlmann2013} that ignores the dynamic structures. The penalty parameter $\lambda_1$ in the Lasso is set to the scaled Lasso parameter $\sqrt{2 \log(p) / n}$.
\end{enumerate}
In all time-varying models, we used the kernel bandwidth $b_n = 0.1$. For the proposed method, we used $\lambda_2 = 1/n$ and $\zeta = 0$. We let $P_{j,t,m}$ be the multiplicity-adjusted p-value for testing the hypothesis $H_{0,j,t}: \beta_j(t) = 0$ for $t \in \varpi = [b_n, 1-b_n]$ in the $m$-th Monte Carlo simulation for $m=1,\cdots,M$. For TV-LDPE, TV-SDL, the proposed method and its non-tv version, we adopted the following performance measures.
\begin{enumerate}
\item The (averaged) false positive rate (FPR) over the interval $\varpi$,
\begin{equation*}
\frac{1}{n(1-2b_n)(p-s)M} \sum_{t \in \varpi} \sum_{j \in \mathcal{S}^\mathsf{c}} \sum_{m=1}^M \vone(P_{j,t,m} \leq \alpha).
\end{equation*}

\item The (averaged) false negative rate (FNR) over the interval $\varpi$,
\begin{equation*}
\frac{1}{n(1-2b_n)sM} \sum_{t \in \varpi}  \sum_{j \in \mathcal{S}} \sum_{m=1}^M \vone(P_{j,t,m} > \alpha).
\end{equation*}

\item The (averaged) FWER over the interval $\varpi$,
\begin{equation*}
\text{FWER} = \frac{1}{n(1-2b_n)M} \sum_{t \in \varpi} \sum_{m=1}^M \vone(\min_{j \in \mathcal{S}^\mathsf{c}} P_{j,t,m} \leq \alpha).
\end{equation*}
\end{enumerate}

For the Lasso-based methods (TV-Lasso and FP-Lasso), the probabilities are replaced by the corresponding indicators of whether or not the estimated coefficients are zero.

\subsection{Empirical results}

For each simulation setup, we report the FPR, RNR, FWER, and the root mean square errors (RMSE) of the estimates. The results are shown in Tables \ref{tab:simulation_n=300_p=300_s=3_1} and \ref{tab:simulation_n=300_p=300_s=3_1_ctd}, from which we make several observations. First, TV-Lasso and the method of \cite{buhlmann2013} do not control the FWER, while the proposed method can control the FWER at the nominal level $\alpha=0.05$ in all setups. Second, the proposed method has uniformly higher power than FP-Lasso, the TV-Lasso tuned to match the FWER with our method. This is probably explained by the bias of the $\ell_1$ regularization in the TV-Lasso. Third, for the design matrix with iid Gaussian entries, TV-LDPE and TV-SDL have comparable performance as our proposed method in terms of the power, while all three methods have the FWER controlled below 0.05. TV-LDPE and TV-SDL are more sensitive to the design matrix than the proposed method; for the Toeplitz design matrix $T$, TV-LDPE and TV-SDL seem to lose control on the FWER. Moreover, the FWER, FNR, and RMSE are larger as the dependence level of the error process grows and as the tail of the error distribution becomes thicker.  Finally, the proposed method is computationally more economical than the competing methods TV-LDPE and TV-SDL. Table \ref{tab:runtime_1} shows the runtimes on an Intel i5-4790K with the Intel MKL linear algebra libraries (software and platform: R 3.2.2 for Windows).

%%%%%%%%%%%%%%%%%%%%%%%%%%%%
\begin{table}[h!]
\centering
\scriptsize
\caption{Simulation results. $n=300, p=300, s=3$.}
\label{tab:simulation_n=300_p=300_s=3_1}
\begin{tabular}{lcccc|cccc}
& \multicolumn{4}{c}{$\vx_i \stackrel{\text{iid}}{\sim} N(\vzero, I_p)$, $e \sim N(\vzero, I_n)$} & \multicolumn{4}{c}{$\vx_i \stackrel{\text{iid}}{\sim} N(\vzero, T)$, $e \sim N(\vzero, I_n)$}\\\toprule
Method & FPR & FNR & FWER & RMSE & FPR & FNR & FWER & RMSE\\\midrule
TV-Lasso & $7.51 \times 10^{-2}$  & 0.0551 & 1 & 0.0537 & $1.55 \times 10^{-1}$ &  0.0684 & 1 & 0.0520\\
FP-Lasso & $1.50 \times 10^{-4}$  & 0.2352 & 0.0344 & 0.1124 & $2.81 \times 10^{-5}$ &  0.5170 & 0.0063 & 0.1318\\
Proposed & $1.50 \times 10^{-4}$  & 0.1889 & 0.0346 & 0.1838 & $2.81 \times 10^{-5}$ & 0.4072 & 0.0063 & 0.1984\\
TV-LDPE & $1.29 \times 10^{-4}$  & 0.1981 & 0.0254 & 0.2652 & $3.77 \times 10^{-4}$ & 0.3615 & 0.0743 & 0.2727\\
TV-SDL & $1.53 \times 10^{-4}$  & 0.1848 & 0.0357 & 0.1316 & $1.08 \times 10^{-3}$ &  0.3006 &  0.2119 & 0.1409\\
Non-TV & $4.89 \times 10^{-1}$ & 0.5100 & 0.4600 & 0.5762 & $2.58 \times 10^{-1}$ &  0.7100 & 0.2400 & 0.6084 \\\bottomrule
\end{tabular}

\begin{tabular}{lcccc|cccc}
& \multicolumn{4}{c}{$\vx_i \stackrel{\text{iid}}{\sim} N(\vzero, I_p)$, $e \sim$ AR(1) with $\varphi = 0.2$} & \multicolumn{4}{c}{$\vx_i \stackrel{\text{iid}}{\sim} N(\vzero, T)$, $e \sim$ AR(1) with $\varphi = 0.2$}\\\toprule
Method & FPR & FNR & FWER & RMSE & FPR & FNR & FWER & RMSE\\\midrule
TV-Lasso & $8.93 \times 10^{-2}$ & 0.0563 & 1 & 0.0533 & $1.51 \times 10^{-1}$ &0.0629 &  1 & 0.0519\\
FP-Lasso & $1.78 \times 10^{-4}$ & 0.2363 & 0.0384 & 0.1099 & $7.21 \times 10^{-5}$ & 0.4709 & 0.0173 & 0.1302\\
Proposed & $1.78 \times 10^{-4}$ & 0.1891 & 0.0376 & 0.1836 & $7.21 \times 10^{-5}$ & 0.3434 & 0.0173 & 0.1920\\
TV-LDPE & $9.85 \times 10^{-5}$ & 0.1995 & 0.0215 & 0.2799 & $3.70 \times 10^{-4}$ & 0.3620 & 0.0843 & 0.2725\\
TV-SDL & $1.97 \times 10^{-4}$ & 0.1883 & 0.0419 & 0.1374 & $1.07 \times 10^{-3}$ & 0.3032 & 0.2165 & 0.1404\\
Non-TV & $4.29 \times 10^{-1}$ & 0.5633 & 0.4100 & 0.5557 & $2.58 \times 10^{-1}$ & 0.6933 & 0.2200 & 0.6088\\\bottomrule
\end{tabular}

\begin{tabular}{lcccc|cccc}
& \multicolumn{4}{c}{$\vx_i \stackrel{\text{iid}}{\sim} N(\vzero, I_p)$, $e \sim$ AR(1) with $\varphi = 0.5$} & \multicolumn{4}{c}{$\vx_i \stackrel{\text{iid}}{\sim} N(\vzero, T)$, $e \sim$ AR(1) with $\varphi = 0.5$}\\\toprule
Method & FPR & FNR & FWER & RMSE & FPR & FNR & FWER & RMSE\\\midrule
TV-Lasso & $7.55 \times 10^{-2}$ & 0.0544 & 1 & 0.0537 & $1.51 \times 10^{-1}$ & 0.0611 & 1 & 0.0518\\
FP-Lasso & $1.93 \times 10^{-4}$ & 0.2402 & 0.0431 & 0.1124 & $9.81 \times 10^{-5}$ & 0.4805 & 0.0222 & 0.1303\\
Proposed & $1.93 \times 10^{-4}$ & 0.1809 & 0.0422 & 0.1836 & $9.81 \times 10^{-5}$ & 0.3347 & 0.0235 & 0.1918\\
TV-LDPE & $9.82 \times 10^{-5}$ & 0.2028 & 0.0229 & 0.2756 & $3.71 \times 10^{-4}$ & 0.3618 & 0.0849 & 0.2659\\
TV-SDL & $1.94 \times 10^{-4}$ & 0.1898 & 0.0425 & 0.1374 & $1.01 \times 10^{-3}$ & 0.2993 & 0.2555 & 0.1439\\
Non-TV & $5.49 \times 10^{-1}$ & 0.4500 & 0.5200 & 0.5314 & $1.70 \times 10^{-1}$ & 0.8300 & 0.1500 & 0.6004 \\\bottomrule
\end{tabular}
\end{table}

\begin{table}[h!]
\centering
\scriptsize
\caption{Simulation results (continued). $n=300, p=300, s=3$.}
\label{tab:simulation_n=300_p=300_s=3_1_ctd}

\begin{tabular}{lcccc|cccc}
& \multicolumn{4}{c}{$\vx_i \stackrel{\text{iid}}{\sim} N(\vzero, I_p)$, $e_i \stackrel{\text{iid}}{\sim} t(3)/\sqrt{3}$} & \multicolumn{4}{c}{$\vx_i \stackrel{\text{iid}}{\sim} N(\vzero, T)$, $e_i \stackrel{\text{iid}}{\sim} t(3)/\sqrt{3}$}\\\toprule
Method & FPR & FNR & FWER & RMSE & FPR & FNR & FWER & RMSE \\\midrule
TV-Lasso & $9.14 \times 10^{-2}$ & 0.0518 & 1 & 0.0539 & $1.57 \times 10^{-1}$ &  0.0605 & 1 & 0.0506\\
FP-Lasso & $1.96 \times 10^{-4}$ & 0.2193 & 0.0398 & 0.1120 & $4.20 \times 10^{-5}$ & 0.4547 & 0.0124 & 0.1209\\
Proposed & $1.96 \times 10^{-4}$ &  0.1708 & 0.0439 & 0.1834 & $4.20 \times 10^{-5}$ & 0.3129 & 0.0125 & 0.1885\\
TV-LDPE & $1.26 \times 10^{-4}$ & 0.2041 & 0.0275 & 0.2659 & $3.62 \times 10^{-4}$ & 0.3043 & 0.0810 & 0.2719\\
TV-SDL & $1.90 \times 10^{-4}$ & 0.1903 & 0.0403 & 0.1321 & $9.54 \times 10^{-4}$ & 0.2814 & 0.1960 & 0.1381\\
Non-TV & $5.16 \times 10^{-1}$ & 0.4800 & 0.4800 & 0.5289 & $2.24 \times 10^{-1}$ & 0.7700 & 0.1500 & 0.5962\\\bottomrule
\end{tabular}

\begin{tabular}{lcccc|cccc}
& \multicolumn{4}{c}{$\vx_i \stackrel{\text{iid}}{\sim} N(\vzero, I_p)$, $e \sim$ LRD with $\varrho = 0.75$} & \multicolumn{4}{c}{$\vx_i \stackrel{\text{iid}}{\sim} N(\vzero, T)$, $e \sim$ LRD with $\varrho = 0.75$}\\\toprule
Method & FP(\%) & FN(\%) & FWER & RMSE & FPR & FNR & FWER & RMSE\\\midrule
TV-Lasso & $7.25 \times 10^{-2}$ & 0.0652 & 1 & 0.0499 & $1.77 \times 10^{-1}$ &  0.0805 &  1 & 0.0556\\
FP-Lasso & $1.60 \times 10^{-4}$ & 0.2496 & 0.0450 & 0.1091 & $1.37 \times 10^{-4}$ &  0.5138 &  0.0229 & 0.1381\\
Proposed & $1.60 \times 10^{-4}$ & 0.1783 & 0.0433 & 0.1806 & $1.37 \times 10^{-4}$ &  0.3376 &  0.0243 & 0.1953\\
TV-LDPE & $1.08 \times 10^{-4}$ & 0.2067 & 0.0238 & 0.2653 & $3.68 \times 10^{-4}$ &  0.3648 &  0.0859 & 0.2691\\
TV-SDL & $2.10 \times 10^{-4}$ & 0.1924 & 0.0501 & 0.1386 & $9.29 \times 10^{-4}$ &  0.3014 & 0.0186 & 0.1448\\
Non-TV & $5.19 \times 10^{-1}$ & 0.4800 & 0.4800 & 0.5304 & $5.49 \times 10^{-1}$ &  0.4500 & 0.4100 & 0.5280 \\\bottomrule
\end{tabular}
\end{table}

\begin{table}[h!]
\centering
\scriptsize
\caption{Runtime per 10 simulations.}
\label{tab:runtime_1}

\vspace{0.5pc}
\begin{tabular}{lr}
\toprule
Method & Runtime (in minutes)\\\midrule
TV-Lasso & 0.5\\
FP-Lasso & 9.5\\
%Proposed (raw p-values only) & 9\\
Proposed & 13\\
TV-LDPE & $>1000$\\
TV-SDL & $>1000$\\
Non-TV & $<0.5$ \\\bottomrule
\end{tabular}
\end{table}

%\noindent The reported runtimes for FP-Lasso includes time spent on divide and conquer $\lambda_1$ and $\alpha$ searches for FPR matching.

\section{Data example: learning brain connectivity}
\label{sec:real_data}

We illustrate our proposed method in an application to model the functional brain connectivity in a Parkinson's disease study. The problem is to construct brain connectivity networks from the resting-state functional magnetic resonance imaging (fMRI) data, where slowly time-varying graphs have implications in modeling brain connectivity networks. Traditional correlation analysis of the resting state blood-oxygen-level-dependent (BOLD) signals of the brain showed considerable temporal variation on small time scales, \cite{changglover2010a,hutchison2013a}. In view of the high spatial resolution of fMRI data, brain networks of subjects at rest are believed to be structurally homogeneous with subtle fluctuations in some, but a small number of, connectivity edges, \cite{kiviniemi2005,hutchison2010a}. A popular approach to learn brain connectivity is the neighborhood selection procedure, \cite{meinshausenbuhlmann2006}. Therefore, high-dimensional TVCM with a small number of nonzero components is a natural approach to study the time-evolving sparse brain connectivity networks, in which the time-varying coefficients reflect the dynamic features of the corresponding edges in the networks. The neighborhood selection approach is an approximation to the full multivariate distributions while ignoring the correlation among the node-wise responses and this may cause certain power loss in finite samples. \cite{meinshausenbuhlmann2006} showed that in terms of variable selection, these two approaches are asymptotically equivalent.

Our data example uses fMRI data collected from a study of patients with Parkinson's disease (PD) and their respective normal controls. PD is typically characterized by deviations in functional connectivity between various regions of the brain. Additionally, the resting state functional connectivity has been shown as a candidate biomarker for PD progression and treatment, where more advanced stages or manifestations of PD are associated with greater deviations from normal connectivity. Each resting state data matrix in our example contains 240 time points and 52 brain regions of interest (ROI). The time points are evenly sampled and the time indices are normalized to [0,1].

The brain connectivity network was constructed using the neighborhood selection procedure. In essence, it is a sequence of time-varying linear regressions enumerating each ROI as the response variable and sparsely regressing on all the other ROIs. Figure \ref{fig:network_normal} and \ref{fig:network_pd} show the estimated graphs of a normal subject and a PD subject at three sequential time points around $t=0.25$ based on the proposed method. Red nodes are ROIs known to be associated with motor control and blue nodes are ROIs either known to be unrelated to motor control or whose functions in humans are not well understood. Different patterns of connectivity in the networks can be found by comparing normal and PD subjects. From the graphs, there are slow changes in the networks over time: most edges are preserved on a small time scale, but there are a few edges evolving over time. For instance, in a PD subject, ROI 1 and ROI 40 are unconnected in the first network but they are connected in the second network and remain connected in the third network. %Generally, when there is substantial activity to study, we tend to observe a more diverse set of edges in healthy subjects and greater connectivity involving the motor-related regions in PD subjects.

\begin{figure}[!ht]
  \centering
  \caption{Connectivity networks in control subject around $t = 0.25$ based on the proposed method.}
  \vspace{1pc}
  \includegraphics[width=400pt]{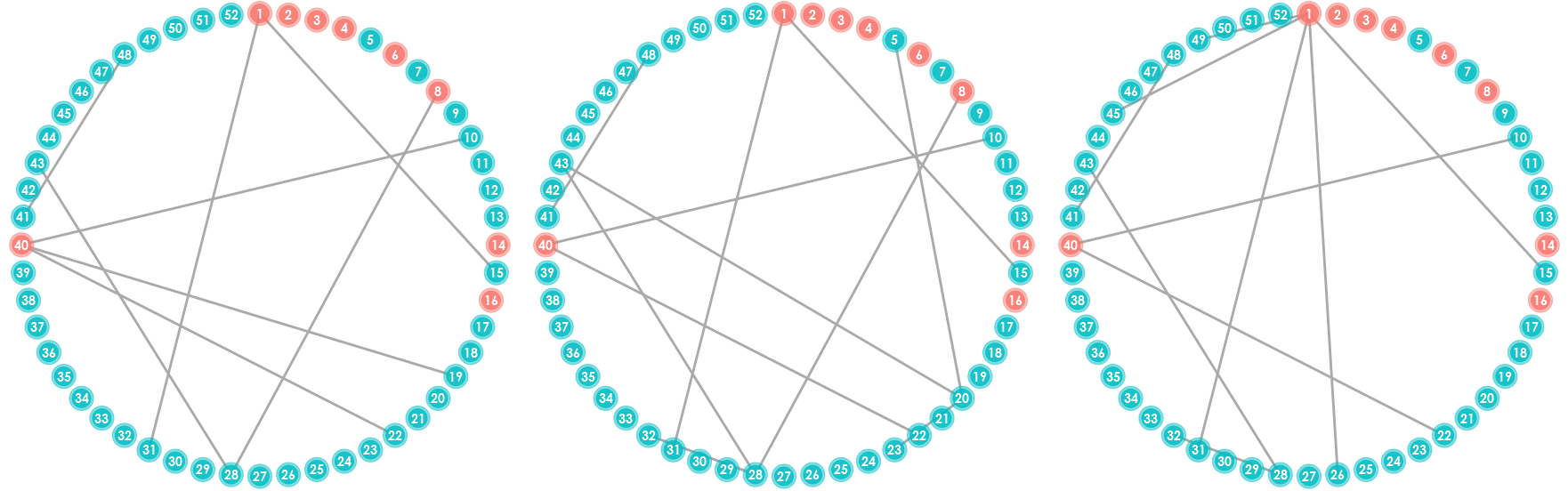}
  \label{fig:network_normal}
  \vspace{1pc}
  \caption{Connectivity networks in Parkinson's Disease subject around $t = 0.25$ based on the proposed method.}
  \vspace{1pc}
  \includegraphics[width=400pt]{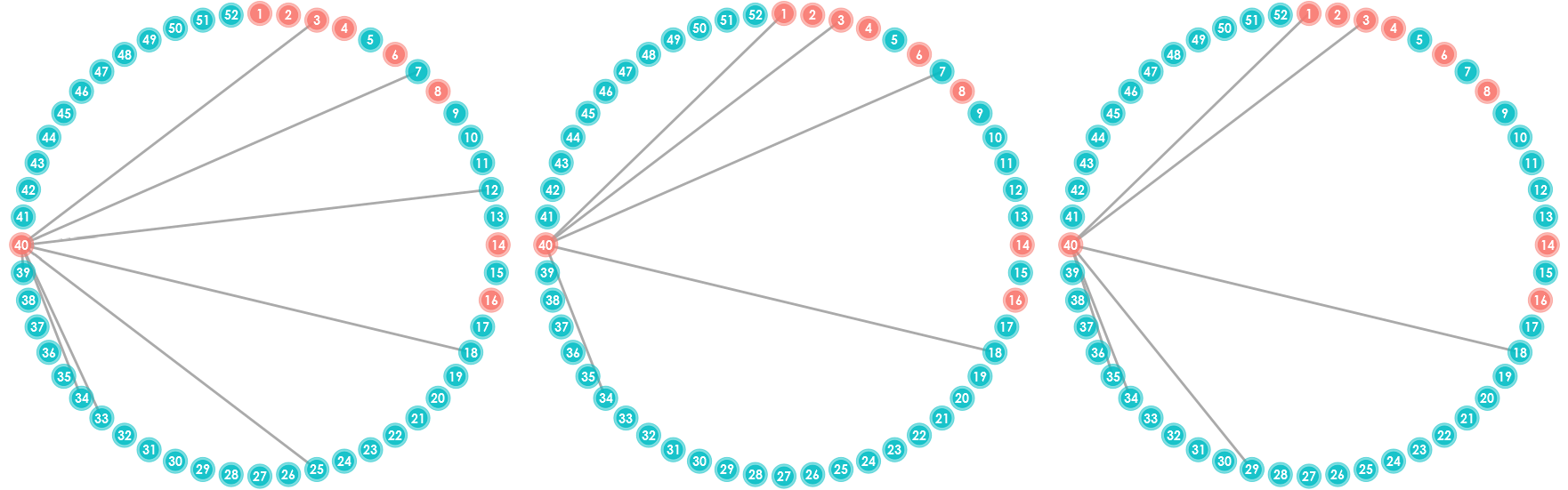}
  \label{fig:network_pd}
\end{figure}

%\begin{figure}[!ht]
%  \centering
%  \label{fig:network}
%  \caption{Connectivity Network in Control Subject around $t = 0.80$}
%  \vspace{1pc}
%  \includegraphics[width=400pt]{images/empty}
%  \vspace{1pc}
%  \caption{Connectivity Network in Parkinson's Subject around $t = 0.50$}
%  \vspace{1pc}
%  \includegraphics[width=400pt]{images/empty}
%\end{figure}

%There are also times when, due to the nature of stringent FWER control and variable selection process, that the estimated connectivity networks may be exceptionally sparse or empty. This serves to reinforce the importance of the time-varying design, particularly for resting state data, since traditional methods which don't study temporal variations have less power to detect differences between subjects when the signals only manifest for short periods of time.

We also plot the connectivity graphs generated by the competing methods TV-LDPE and TV-SDL; see Figure \ref{fig:network_ldpe} and \ref{fig:network_sdl}. These graphs are denser than with the proposed method and are harder to interpret. Besides, as commented in the simulation studies, TV-LDPE and TV-SDL are more computationally expensive. The method of \cite{buhlmann2013} cannot be used to capture the dynamic features of brain connectivity networks. 

\begin{figure}[!ht]
  \centering
  \caption{A connectivity network based on TV-LDPE.}
  \vspace{1pc}
  \includegraphics[width=400pt]{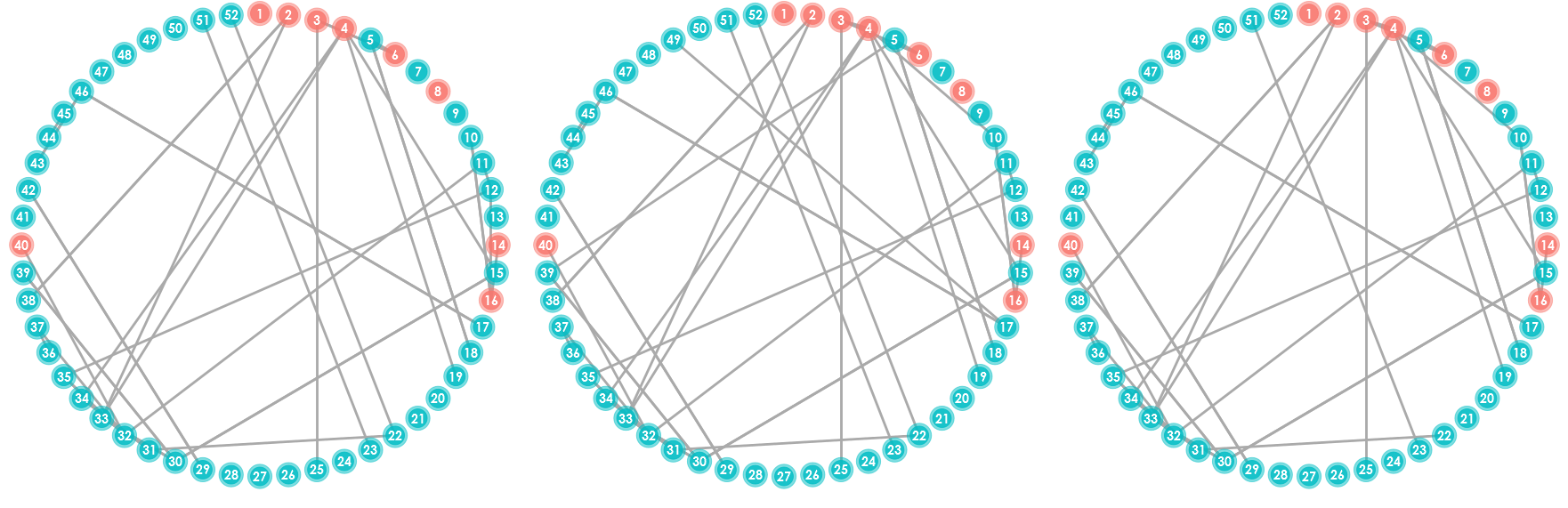}
  \label{fig:network_ldpe}
  \vspace{1pc}
  \caption{A connectivity network based on TV-SDL.}
  \label{fig:network_sdl}
  \vspace{1pc}
  \includegraphics[width=400pt]{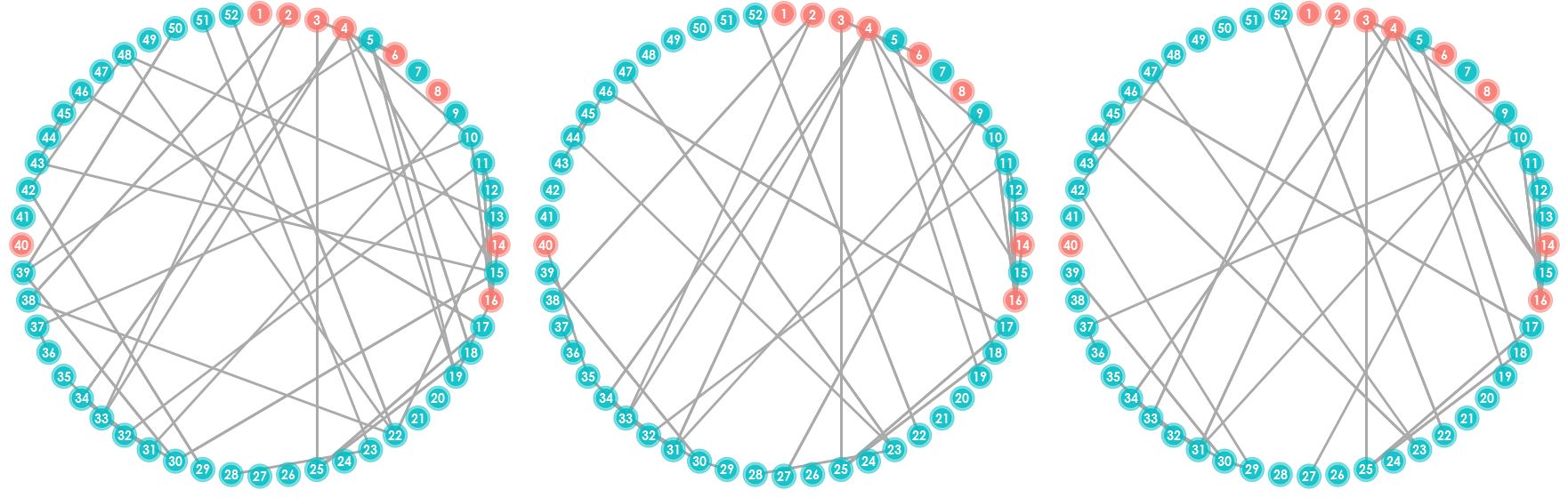}
\end{figure}

\section{Discussions}
\label{sec:discussion}

This paper presents a pointwise inference algorithm for high-dimensional TVCM that can asymptotically control the FWER. Based on the current work, an interesting improvement would be to study simultaneous inference. Construction of the simultaneous confidence band (SCB) for the time-varying coefficients is useful for testing their parametric forms in high dimensions. This is a more challenging topic, which requires substantial additional work and probability tools that are beyond the scope of this paper.

Our brain connectivity application is a subject-by-subject analysis. To perform the group analysis on the population level, a hierarchical linear model is more appropriate. When $p$ is fixed, the linear mixed-effects model is widely used in performing the multi-level group analysis in fMRI, \cite{beckmanjenkinsonsmith2003,lindquist2008,skup2010}. The reason is that the generalized least squares (GLS) estimator of a two-level model is inferentially equivalent to the GLS estimator of the corresponding single-level model, provided that the second-level covariance is the sum of the group covariance and the covariance of the first-level estimate, \cite{beckmanjenkinsonsmith2003}. Extension of the testing problem on the population parameters based on the ridge and Lasso estimates when $p \to \infty$ is another interesting future research topic.

\section{Proof}
\label{sec:proof}

\begin{proof}[Proof of Theorem \ref{thm:representation}]
Observe that ${\cal X}_t \mbf\beta(t) = {\cal X}_t \mbf\theta(t)$ since $\mbf\theta(t) = P_{{\cal R}_t}\mbf\beta(t)$. Using the closed-form formulae for the tv-ridge estimator (\ref{eqn:tv-ridge}) and by (\ref{eqn:taylor}), we have
\begin{equation}
\label{eqn:bias_ridge_total}
\text{bias}(\tilde{\mbf\theta}(t)) = \E(\tilde{\mbf\theta}(t)) - \mbf\theta(t) = ({\cal X}_t^\top {\cal X}_t + \lambda_2 I_p)^{-1} {\cal X}_t^\top [ {\cal X}_t \mbf\theta(t) + M_t {\cal X}_t \mbf\beta'(t) + {\cal X}_t \mbf\xi ] - \mbf\theta(t),
\end{equation}
where $|\mbf\xi|_\infty \le C_0 b_n^2 /2$ and $|\mbf\xi|_0 \le s$ almost surely, $t \in \varpi$. First, we bound the shrinkage bias of the tv-ridge estimator. By the argument in Section 3 of \cite{shaodeng2012}, we can show that
\begin{equation*}
({\cal X}_t^\top {\cal X}_t + \lambda_2 I_p)^{-1} {\cal X}_t^\top {\cal X}_t \mbf\theta(t) - \mbf\theta(t) = - Q (\lambda_2^{-1} D^2 + I_r)^{-1} Q^\top \mbf\theta(t).
\end{equation*}
It follows from Lemma \ref{lem:unweighted-2-weighted} that
\begin{eqnarray}
\label{eqn:tvridge-bias-shrinkage}
| Q (\lambda_2^{-1} D^2 + I_r)^{-1} Q^\top \mbf\theta(t) |_2 &\le& {|\mbf\theta(t)|_2 \over \rho_{\min}(\lambda_2^{-1} D^2 + I_r) } \\
\nonumber
&=& \left( {\lambda_2 \over \lambda_2 + \min_{j \le r} d_j^2} \right) |\mbf\theta(t)|_2 \le {\lambda_2 |\mbf\theta(t)|_2  \over \rho_{\min \neq 0}(\hat\Sigma_t)} \le {\lambda_2 |\mbf\theta(t)|_2  \over |N_t| \underline{w}_t \varepsilon_0^2},
\end{eqnarray}
where $d_j^2 = \rho_j(\hat\Sigma_t), j = 1, \cdots, r$. Next, we deal with the non-stationary bias of the tv-ridge estimator (\ref{eqn:bias_ridge_total}) by a similar argument for (\ref{eqn:tvridge-bias-shrinkage}). Indeed, let $Q_\perp$ be the orthogonal complement of $Q$ such that $Q_\perp^\top Q_\perp = I_{p-r}$ and $Q_\perp^\top Q = \vzero_{(p-r) \times r}$. Denote $\Gamma = [Q; Q_{\perp}]$. Then, $\Gamma \Gamma^\top = \Gamma^\top \Gamma = I_p$. By the SVD of ${\cal X}_t$, (\ref{eqn:svd_X}), we have
\begin{align*}
& ({\cal X}_t^\top {\cal X}_t + \lambda_2 I_p)^{-1} {\cal X}_t^\top M_t {\cal X}_t {\mbf\beta}'(t) = \Gamma \left( \Gamma^\top (Q D^2 Q^\top + \lambda_2 I_p) \Gamma \right)^{-1} \Gamma^\top {\cal X}_t^\top M_t {\cal X}_t {\mbf\beta}'(t) \\
& \qquad  = [Q; Q^\perp] \left( \left[\begin{array}{c} Q^\top \\ Q_{\perp}^\top \end{array} \right] (Q D^2 Q^\top + \lambda_2 I_p)  [Q; Q_\perp] \right)^{-1} \left[\begin{array}{c} Q^\top \\ Q_{\perp}^\top \end{array} \right]  Q D P^\top M_t {\cal X}_t {\mbf\beta}'(t) \\
& \qquad = [Q; Q^\perp] \left( \begin{array}{cc} (D^2 + \lambda_2 I_r)^{-1} & \vzero \\ \vzero & \lambda_2^{-1} I_{p-r}  \end{array} \right)  \left[\begin{array}{c} D P^\top M_t  {\cal X}_t {\mbf\beta}'(t)  \\ \vzero \end{array} \right] \\
& \qquad = Q ( D + \lambda_2 D^{-1} )^{-1} P^\top M_t  {\cal X}_t {\mbf\beta}'(t).
\end{align*} 
Hence, by Lemma \ref{lem:unweighted-2-weighted} we have
\begin{equation*}
|({\cal X}_t^\top {\cal X}_t + \lambda_2 I_p)^{-1} {\cal X}_t^\top M_t {\cal X}_t \mbf\beta'(t)|_2 \le {  b_n |{\cal X}_t \mbf\beta'(t)|_2 \over \rho_{\min}(D + \lambda_2 D^{-1}) } \le {C_0 b_n (s |N_t| \overline{w}_t)^{1/2} \varepsilon_0^{-1} \over \min_{j \le r} (d_j + \lambda_2 / d_j) },
\end{equation*}
where $\overline{w}_t = \sup_{i \in N_t} w(i,t)$. Since $\lambda_2 = o(1)$ and $d_j \ge (|N_t| \underline{w}_t)^{1/2} \varepsilon_0$, the denominator of last expression is lower bounded by $[ (|N_t| \underline{w}_t)^{1/2}\varepsilon_0 + \lambda_2 / ( (|N_t| \underline{w}_t)^{1/2} \varepsilon_0)]$ for large enough $n$. Therefore, we have
\begin{equation}
\label{eqn:tvridge-bias-nonstationary}
| ({\cal X}_t^\top {\cal X}_t + \lambda_2 I_p)^{-1} {\cal X}_t^\top M_t {\cal X}_t \mbf\beta'(t) |_2 \le {C_0 b_n (s |N_t| \overline{w}_t)^{1/2} \over  (|N_t| \underline{w}_t)^{1/2} \varepsilon_0^2 } \le {C_0 b_n  s^{1/2} \over C \varepsilon_0^2 }.
\end{equation}
Similarly, an upper bound for the remainder term of (\ref{eqn:bias_ridge_total}) can be established. We have
\begin{equation}
\label{eqn:tvridge-bias-nonstationary-remainder}
| ({\cal X}_t^\top {\cal X}_t + \lambda_2 I_p)^{-1} {\cal X}_t^\top M_t {\cal X}_t \mbf\xi |_2 \le {C_0 b_n^2 s^{1/2} \over 2 C \varepsilon_0^2}, \quad \text{for almost surely } t \in \varpi.
\end{equation}
In addition, $\tilde{\mbf\theta}(t) - \E[\tilde{\mbf\theta}(t)] = ({\cal X}_t^\top {\cal X}_t + \lambda_2 I_p)^{-1} {\cal X}_t^\top {\cal E}_t$ is the stochastic part of the tv-ridge estimator. Since the $e_i \sim N(0, \sigma^2 I_n)$ are iid, ${\cal E}_t \sim N(\vzero, \sigma^2 W_t)$. Hence, $\tilde{\mbf\theta}(t) - \E[\tilde{\mbf\theta}(t)] \sim N(\vzero, \Omega(\lambda_2))$, where $\Omega(\lambda)$ is defined in (\ref{eqn:Omega}), and thus
\begin{equation}
\Var(\tilde{\theta}_j(t)) = \Omega_{jj}(\lambda_2) \ge \Omega_{\min}(\lambda_2).
\end{equation}
Now, we consider the initial tv-lasso estimator. By Lemma \ref{lem:tv-lasso-rate},
\begin{equation}
\label{eqn:tvlasso-bias-proj}
|\tilde{\mbf\beta}(t) - {\mbf\beta}(t)|_1 \le 4 \phi_0^{-2} \lambda_1 s.
\end{equation}
Then, (\ref{eqn:decomposition}), (\ref{eqn:decomposition_stochastic}), and (\ref{eqn:decomposition_bias}) follow by assembling (\ref{eqn:tvridge-bias-shrinkage}), (\ref{eqn:tvridge-bias-nonstationary}), (\ref{eqn:tvridge-bias-nonstationary-remainder}), and (\ref{eqn:tvlasso-bias-proj}) into (\ref{eqn:proposed-estimator}),
\begin{equation*}
\hat{\mbf\beta}(t) = {\mbf\beta}(t) + \text{bias}(\tilde{\mbf\theta}(t)) + \{ \tilde{\mbf\theta}(t) - \E[\tilde{\mbf\theta}(t)] \} - \{(P_{{\cal R}_t} - I_p) [\tilde{\mbf\beta}(t) - \mbf\beta(t)]\}.
\end{equation*}
The marginal representation (\ref{eqn:decomposition_marginal}) and (\ref{eqn:gamma_marginal_beta_j=0}) follow from similar arguments by noting that $B_j(t) = \sum_{k \neq j} (P_{{\cal R}_t})_{jk} \beta_k(t)$ under $H_{0,j,t}$.
\end{proof}

\begin{proof}[Proof of Theorem \ref{thm:representation_GP}]
The proof of Theorem \ref{thm:representation_GP} is similar to that of Theorem \ref{thm:representation} so we only highlight the difference involving the error process. First, $\Cov({\cal E}_t)=W_t^{1/2}\Sigma_{e,t} W_t^{1/2}$. Second, instead of using (\ref{eqn:concentration_iid_gaussian}) in proving Lemma \ref{lem:tv-lasso-rate}, we use Lemma \ref{lem:concentration_GP} to get for all $\lambda>0$
\begin{equation*}
\Prob \left( \max_{j \le p} \left| \sum_{i \in N_t} w(i,t) X_{ij} e_i \right| \ge \lambda \right) \le 2 p \exp\left(-{\lambda^2 \over 2 L_{t,2}^2 |\va|_1^2 \sigma^2} \right) \quad \text{if } \varrho > 1,
\end{equation*}
\begin{equation*}
\Prob \left( \max_{j \le p} \left| \sum_{i \in N_t} w(i,t) X_{ij} e_i \right| \ge \lambda \right) \le 2 p \exp\left(-{C_\varrho \lambda^2 \over L_{t,2}^2 n^{2(1-\varrho)} \sigma^2 K^2} \right) \quad \text{if } \varrho \in (1/2,1).
\end{equation*}
\end{proof}

\begin{proof}[Proof of Theorem \ref{thm:non-gaussian}]
The proof essentially follows the lines of that of Theorem \ref{thm:representation}, but with differences in requiring a larger penalty parameter $\lambda_1$ of the tv-Lasso. First, by the Nagaev inequality \cite{nagaev1979a}, we have for any $\varepsilon > 0$,
\begin{equation*}
\Prob \left( \max_{j \le p} \left| \sum_{i \in N_t} w(i,t) X_{ij} e_i \right| \ge  \sigma L_{t,2} \varepsilon \right) \le (1+2/q)^q \kappa_q {p \mu_{n,q} \over (\sigma L_{t,2} \varepsilon)^q} + 2  p\exp\left(-{c_q \varepsilon^2}\right),
\end{equation*}
where $c_q =  2 e^{-q} (q+2)^{-2}$ and $\kappa_q$ is the $q$-th absolute moment of $e_1$. Then, choosing 
\begin{equation*}
\varepsilon = C_q \max \left\{ { (p \mu_{n,q})^{1/q} \over \sigma L_{t,2}}, \; (\log{p})^{1/2} \right\} \quad \text{for large enough  } C_q > 0,
\end{equation*}
we have $\max_{j \le p} | \sum_{i \in N_t} w(i,t) X_{ij} e_i | = O_\Prob(\lambda_0)$. Second, let $\Xi =  ({\cal X}_t^\top {\cal X}_t + \lambda_2 I_p)^{-1} {\cal X}_t^\top W_t^{1/2}$ and ${\cal E}_t^\diamond = (e_i)_{i \in N_t}^\top$. Recall that $\tilde{\mbf\theta}(t) - \E[\tilde{\mbf\theta}(t)] =\Xi {\cal E}_t^\diamond$. By the Gaussian approximation \cite[Theorem B]{shao1995}, there exist iid Gaussian random variables $g_i \sim N(0, \sigma^2 \xi_{ji}^2)$ defined on a possibly richer probability space such that for every $t > 0$
\begin{equation*}
\Prob \left( \left| \tilde{\theta}_j(t) - \E[\tilde{\theta}_j(t)] - \sum_{i \in N_t} g_i \right| \ge t \right) \le (C q)^q {\sum_{i \in N_t} \E|\xi_{ji} e_i|^q \over t^q}.
\end{equation*}
Thus, it follows that $\tilde{\theta}_j(t) - \E[\tilde{\theta}_j(t)]  = N(0, \Omega_{jj}(\lambda_2)) + O_{\Prob}( |\mbf\xi_j|_q )$. As $\Omega_{jj}(\lambda_2) = \sigma^2 |\mbf\xi_j|_2^2$, the proof is complete for by assumption, $|\mbf\xi_j|_q = o(|\mbf\xi_j|_2)$.
\end{proof}

%%%%%%%%%%%%%%%%%%%%%%%%%%%%%%%%%%%%%%%%%%%%%%%%%%%%%%%%%%%%%%%%%%%%%%%%%%%%%%%%%%%%%%%%%%%%%%%%%%%%%%%%%%%%%%%%%%%%%%%%%%%%
\vskip 14pt
\noindent {\large\bf Supplementary Material}

The supplementary material contains additional technical lemmas and discusses some implementation issues.
\par
%%%%%%%%%%%%%%%%%%%%%%%%%%%%%%%%%%%%%%%%%%%%%%%%%%%%%%%%%%%%%%%%%%%%%%%%%%%%%%%%%%%%%%%%%%%%%%%%%%%%%%%%%%%%%%%%%%%%%%%%%%%%
\vskip 14pt
\noindent {\large\bf Acknowledgements}

We would like to thank an anonymous referee, an Associate Editor, and the Co-Editor for their many helpful comments that led to improvements of this paper. This research was partially supported by NSF DMS-1404891 and UIUC Research Board Award RB15004. We thank Aiping Liu (University of British Columbia) for providing the Parkinson's Disease data. The high-performance computing of this work was completed on the Illinois Campus Cluster Program at UIUC.\par

%%%%%%%%%%%%%%%%%%%%%%%%%%%%%%%%%%%%%%%%%%%%%%%%%%%%%%%%%%%%%%%%%%%%%%%%%%%%%%%%%%%%%%%%%%%%%%%%%%%%%%%%%%%%%%%%%%%%%%%%%%%%

\bibliographystyle{plain}
\bibliography{tvcm}

\begin{thebibliography}{10}

\bibitem{beckmanjenkinsonsmith2003}
Christian~F. Beckman, Mark Jenkinson, and Stephen~M. Smith.
\newblock General multilevel linear modeling for group analysis in fmri.
\newblock {\em NeuroImage}, 20:1052--1063, 2003.

\bibitem{bickelritovtsybakov2009}
Peter Bickel, Ya'acov Ritov, and Alexandre Tsybakov.
\newblock Simultaneous analysis of lasso and dantzig selector.
\newblock {\em Annals of Statistics}, 37(4):1705--1732, 2009.

\bibitem{bickellevina2008b}
Peter~J. Bickel and Elizaveta Levina.
\newblock {Regularized Estimation of Large Covariance Matrices}.
\newblock {\em The Annals of Statistics}, 36(1):199--227, 2008.

\bibitem{blm2013}
St\'ephane Boucheron, G\'abor Lugosi, and Pascal Massart.
\newblock {\em Concentration Inequalities: A Nonasymptotic Theory of
  Independence}.
\newblock Oxford, 2013.

\bibitem{buhlmann2013}
Peter B\"uhlmann.
\newblock {Statistical significance in high-dimensional linear models}.
\newblock {\em Bernoulli}, 19(4):1212--1242, 2013.

\bibitem{buhlmannvandegeer2011}
Peter B\"uhlmann and Sara van~de Geer.
\newblock {\em {Statistics for High-Dimensional Data: Methods, Theory and
  Applications}}.
\newblock Springer Series in Statistics, 2011.

\bibitem{caizhangzhou2010a}
Tony Cai, Cun-Hui Zhang, and Harrison Zhou.
\newblock {Optimal Rates of Convergence for Covariance Matrix Estimation}.
\newblock {\em The Annals of Statistics}, 38(4):2118--2144, 2010.

\bibitem{cai_zw2007}
Zongwu Cai.
\newblock Trending time-varying coefficient time series models with serially
  correlated errors.
\newblock {\em Journal of Econometrics}, 136(1):163--188, 2007.

\bibitem{changglover2010a}
Catie Chang and Gary~H. Glover.
\newblock {Time-frequency dynamics of resting-state brain connectivity measured
  with fMRI}.
\newblock {\em NeuroImage}, 50:81--98, 2010.

\bibitem{chenxuwu2013}
Xiaohui Chen, Mengyu Xu, and Wei~Biao Wu.
\newblock {Covariance and precision matrix estimation for high-dimensional time
  series}.
\newblock {\em Annals of Statistics}, 41(6):2994--3021, 2013.

\bibitem{chenxuwu2016}
Xiaohui Chen, Mengyu Xu, and Wei~Biao Wu.
\newblock {Regularized estimation of linear functionals of precision matrices
  for high-dimensional time series}.
\newblock {\em IEEE Transactions on Signal Processing}, 64(24):6459--6470,
  2016.

\bibitem{clevelandgrosseshyu1991}
W.~S. Cleveland, E.~Grosse, and W.~M. Shyu.
\newblock {Local regression models.}
\newblock {\em In Statistical Models in S (Chambers, J. M. and Hastie, T. J.,
  eds) Wadsworth \& Brooks, Pacific Grove.}, pages 309--376, 1991.

\bibitem{fanha2016}
Jianqing Fan and Xu~Han.
\newblock Estimation of the false discovery proportion with unknown dependence.
\newblock {\em Journal of the Royal Statistical Society: Series B (Statistical
  Methodology)}, 2016.

\bibitem{fanhangu2012}
Jianqing Fan, Xu~Han, and Weijie Gu.
\newblock Estimating false discovery proportion under arbitrary covariance
  dependence.
\newblock {\em Journal of American Statistical Association},
  107(499):1019--1035, 2012.

\bibitem{fanzhang1999}
Jianqing Fan and Wenyang Zhang.
\newblock Statistical estimation in varying coefficient models.
\newblock {\em Annals of Statistics}, 27(5):1491--1518, 1999.

\bibitem{hooveretal1998}
Donald~R. Hoover, John~A. Rice, Colin~O Wu, and Li-Ping Yang.
\newblock Nonparametric smoothing estimates of time-varying coefficient models
  with longitudinal data.
\newblock {\em Biometrika}, 85(4):809--822, 1998.

\bibitem{hutchison2010a}
M.~Hutchison, Mirsattari S., C.~Jones, J.~Gati, and L.~Leugn.
\newblock {Functional networks in the anesthetized rat brain revealed by
  independent component analysis of resting-state fMRI}.
\newblock {\em J. Neurophysiol}, 103:3398--3406, 2010.

\bibitem{hutchison2013a}
M.~Hutchison, T.~Womelsdorf, J.~Gati, S.~Everling, and R.~Menon.
\newblock {Resting-state networks show dynamic functional connectivity in awake
  humans and anesthetized macaques}.
\newblock {\em Human Brain Mapping}, 34:2154--2177, 2013.

\bibitem{javanmardmontanari2014a}
Adel Javanmard and Andrea Montanari.
\newblock Confidence inter hypothesis testing in high-dimensional regression
  under the gaussian random design model: Asymptotic theory.
\newblock {\em IEEE Transactions on Information Theory}, 60(10):6522--6554,
  2014.

\bibitem{kiviniemi2005}
V.~Kiviniemi, H.~Haanp{\"a}{\"a}, J-H Kantola, Jauhiainen J.,
  Vainionp{\"a}{\"a} V., S.~Alahuhta, and O.~Tervonen.
\newblock {Midazolam sedation increases fluctuation and synchrony of the
  resting brain BOLD signal}.
\newblock {\em Magn Reson Imaging}, 23:531--537, 2005.

\bibitem{lindquist2008}
Martin Lindquist.
\newblock The statistical analysis of fmri data.
\newblock {\em Statistical Science}, 23(4):439--464, 2008.

\bibitem{mcmurrypolitis2010}
T.~L. McMurry and D.~N. Politis.
\newblock {Banded and tapered estimates for autocovariance matrices and the
  linear process bootstrap.}
\newblock {\em J. Time Ser. Anal.}, 31:471--482, 2010.

\bibitem{meinshausenbuhlmann2006}
Nicolai Meinshausen and Peter B\"uhlmann.
\newblock {High-dimensional graphs and variable selection with the Lasso}.
\newblock {\em Annals of Statistics}, 34(3):1049--1579, 2006.

\bibitem{nagaev1979a}
S.V. Nagaev.
\newblock {Large deviations of sums of independent random variables}.
\newblock {\em Annals of Probability}, 7(5):745--789, 1979.

\bibitem{orbeferreirarodriguez-poo2005}
Susan Orbe, Eva Ferreira, and Juan Rodriguez-Poo.
\newblock Nonparametric estimation of time varying parameters under shape
  restrictions.
\newblock {\em Journal of Econometrics}, 126(1):53--77, 2005.

\bibitem{resnick1987a}
Sidney~I. Resnick.
\newblock {\em Extreme Values, Regular Variation, and Point Processes}.
\newblock Applied Probability. Springer-Verlag, 1987.

\bibitem{robinson1989}
Peter~M. Robinson.
\newblock {\em Nonparametric estimation of time-varying parameters}.
\newblock { In Statistical Analysis and Forecasting of Economic Structural
  Change (P. Hackl, ed.) Berlin: Springer.}, 1989.

\bibitem{shaodeng2012}
Jun Shao and Xinwei Deng.
\newblock {Estimation in high-dimensional linear models with deterministic
  design matrices}.
\newblock {\em Annals of Statistics}, 40(2):812--831, 2012.

\bibitem{shao1995}
Qi-Man Shao.
\newblock Strong approximation theorems for independent random variables and
  their applications.
\newblock {\em Journal of Multivariate Analysis}, 52:107--130, 1995.

\bibitem{skup2010}
Martha Skup.
\newblock Longitudinal fmri analysis: a review of methods.
\newblock {\em Statistics and Its Interface}, 3(2):235--252, 2010.

\bibitem{ruiyizou2014}
Rui Song, Feng Yi, and Hui Zou.
\newblock On varying-coefficient independence screening for high-dimensional
  varying-coefficient models.
\newblock {\em Statistica Sinica}, 24(4):1735--1752, 2014.

\bibitem{sunzhang2012}
Tingni Sun and Chun-Hui Zhang.
\newblock {Scaled sparse linear regression}.
\newblock {\em Biometrika}, 99:879--898, 2012.

\bibitem{vandegeerbuhlmann2009}
Sara van~de Geer and Peter B\"uhlmann.
\newblock {On the conditions used to prove oracle results for the Lasso}.
\newblock {\em Electronic Journal of Statistics}, 3:1360--1392, 2009.

\bibitem{wangzhang1992}
Boying Wang and Fuzhen Zhang.
\newblock {Some inequalities for the eigenvalues of the product of positive
  semidefinite Hermitian matrices}.
\newblock {\em Linear Algebra and Its Applications}, 160:113--118, 1992.

\bibitem{weihuangli2011}
Fengrong Wei, Jian Huang, and Hongzhe Li.
\newblock Variable selection and estimation in high-dimensional
  varying-coefficient models.
\newblock {\em Statistica Sinica}, 21:1515--1540, 2011.

\bibitem{xuequ2012}
Lan Xue and Annie Qu.
\newblock Variable selection in high-dimensional varying coefficient models
  with global optimality.
\newblock {\em Journal of Machine Learning Research}, 13:1973--1998, 2012.

\bibitem{zhangzhang2013}
Cun-Hui Zhang and Stephanie~S. Zhang.
\newblock Confidence intervals for low dimensional parameters in high
  dimensional linear models.
\newblock {\em Journal of the Royal Statistical Society: Series B (Statistical
  Methodology)}, 71(1):217--242, 2013.

\bibitem{zhangwu2012a}
Ting Zhang and Wei~Biao Wu.
\newblock Inference of time-varying regression models.
\newblock {\em Annals of Statistics}, 40(3):1376--1402, 2012.

\bibitem{zhangleesong2002}
Wenyang Zhang, Sik-Yum Lee, and Xinyuan Song.
\newblock Local polynomial fitting in semivarying coefficient model.
\newblock {\em Journal of Multivariate Analysis}, 82:166--188, 2002.

\bibitem{zhoulaffertywasserman2010}
Shuheng Zhou, John Lafferty, and Larry Wasserman.
\newblock {Time-varying undirected graphs}.
\newblock {\em Machine Learning}, 80(2-3):295--319, 2010.

\bibitem{zhouwu2010}
Zhou Zhou and Wei~Biao Wu.
\newblock {Simultaneous inference of linear models with time varying
  coefficients}.
\newblock {\em Journal of the Royal Statistical Society}, 72(4):513--531, 2010.

\end{thebibliography}

%\begin{thebibliography}{11}
%\expandafter\ifx\csname
%natexlab\endcsname\relax\def\natexlab#1{#1}\fi
%\expandafter\ifx\csname url\endcsname\relax
%  \def\url#1{\texttt{#1}}\fi
%\expandafter\ifx\csname urlprefix\endcsname\relax\def\urlprefix{URL
%}\fi
%\bibitem[Antoine and Renault(2012)]{1} Antoine, B. and Renault, E. (2012). Efficient minimum distance
%estimation with multiple rates of convergence. \textit{J. Economet.} \textbf{170}, 350-367.
%\bibitem[Shao and Tu(1995)]{2}
%Shao, J. and Tu, D. (1995). \textit{The Jackknife and Bootstrap}. Springer-Verlag, New York.
%\bibitem[Zhu et al.(2008)]{3} Zhu, H., Ibrahim, J. G., Tang, N. S. and Zhang, H. (2008).
%Diagnostic measures for empirical likelihood of general estimating equations.
%{\it Biometrika} {\bf 95}, 489-507.
%\bibitem[Zhou, Wan and Wang(2008)]{4} Zhou, Y., Wan, A. T. and Wang, X. (2008). Estimating equations inference with missing data.
%{\it J. Amer. Statist. Assoc.} {\bf 103}, 1187-1199.
%
%
%%%%%%%%%%%%%%%%%%%%%%%%%%%%%%%%%%%%%%%%%%%%%%%%%%%%%%%%%%%%%%%%%%%%%%%%%%%%%%%%%%%%%%%%%%%%%%%%%%%%%%%%%%%%%%%%%%%%%%%%%%%%%
%\end{thebibliography}

\vskip .65cm
\noindent
Department of Statistics, 725 S Wright Street, Champaign, IL USA 61820
\vskip 2pt
\noindent
E-mail: xhchen@illinois.edu, he3@illinois.edu
\vskip 2pt

% \vskip .3cm
%\centerline{(Received ???? 20??; accepted ???? 20??)}\par

\section*{Supplementary Material}
\label{sec:proofs}

\subsection*{Additional technical lemmas}

\begin{lem}
\label{lem:diag_weighted_spectral_bound}
Let $X$ be an $n\times p$ matrix and $D = \diag(d_1,\cdots,d_n)$ with $|d_i| \le b$ and $b \ge 0$. Then
$$
\rho_{\max}(X^\top DX, s) \le 2 b \rho_{\max}(X^\top X, s).
$$
If $d_i \in [0,b]$, then $\rho_{\max}(X^\top DX, s) \le b \rho_{\max}(X^\top X, s).$
\end{lem}

\begin{proof}
Let ${\cal A}_s = \{\va \in \mathbb{R}^p : |\va|_2 \le 1, |\va|_0 \le s\}$. Write $d_i=d_i^+-d_i^-$, where $d_i^+ = \max(d_i,0)$ and $d_i^-=\max(-d_i,0)$ are the positive and negative parts, respectively. By definition
\begin{eqnarray*}
\rho_{\max}(X^\top D X, s) &=& \max_{\va \in {\cal A}_s} |\va^\top X^\top D X \va| = \max_{\va \in {\cal A}_s} |\tr(D (X\va \va^\top X^\top))| \\
&=&  \max_{\va \in {\cal A}_s} \left| \sum_{i=1}^n (d_i^+ - d_i^-) (X\va \va^\top X^\top)_{ii} \right| \le 2 b \max_{\va \in {\cal A}_s} \sum_{i=1}^n (X\va \va^\top X^\top)_{ii} \\
&=& 2 b \max_{\va \in {\cal A}_s} \tr(X\va \va^\top X^\top) = 2 b \max_{\va \in {\cal A}_s} \va^\top X^\top X \va = 2 b \rho_{\max}(X^\top X, s),
\end{eqnarray*}
because $X^\top X$ is nonnegative definite. The second claim follows from the same lines with $d_i^-=0$.
\end{proof}

\begin{lem}
\label{lem:unweighted-2-weighted}
Let $t \in \varpi$ and $\hat\Sigma_t$ be the kernel smoothed sample covariance at time $t$ and $\hat\Sigma_t^\diamond = {{\cal X}_t^\diamond}^\top {{\cal X}_t^\diamond}$. Suppose that ${\cal X}_t^\diamond$ has full row rank. Assume further (\ref{eqn:restricted_min_eigenvalue_Gram}), (\ref{eqn:max_eigenvalue_Gram}) and assumption 6 hold, then we have
\begin{eqnarray}
\label{eqn:unweighted-2-weighted}
\rho_{\min \neq 0}(\hat\Sigma_t) &\ge& |N_t| {\underline{w}_t \varepsilon_0^2}, \\
\label{eqn:unweighted-2-weighted_restricted_upper_bound}
\rho_{\max}(\hat\Sigma_t,s) &\le& |N_t| {\overline{w}_t \varepsilon_0^{-2}}.
\end{eqnarray}
\end{lem}

\begin{proof}
Since ${\cal X}_t^\diamond$ is of full row rank, $r = |N_t|$. Note that ${\cal X}_t = (|N_t| W_t)^{1/2} {\cal X}_t^\diamond$, $\rho_i(\hat\Sigma_t) = \sigma_i^2({\cal X}_t)$ and $\rho_i(\hat\Sigma_t^\diamond) = \sigma_i^2({\cal X}_t^\diamond)$. By the generalized Marshall-Olkin inequality, see e.g. \cite[Theorem 4]{wangzhang1992}, assumption 6 and (\ref{eqn:restricted_min_eigenvalue_Gram}), we have
\begin{eqnarray*}
\rho_{\min\neq0}(\hat\Sigma_t) &=& \rho_{\min}(\calX_t \calX_t^\top) = |N_t| \rho_{\min}(W_t^{1/2} \calX_t^\diamond {\calX_t^\diamond}^\top W_t^{1/2}) \\
&=& |N_t| \rho_{\min}(\calX_t^\diamond {\calX_t^\diamond}^\top W_t) \ge |N_t| \rho_{\min}(W_t) \rho_{\min}(\calX_t^\diamond {\calX_t^\diamond}^\top) \ge |N_t| {\underline{w}_t \varepsilon_0^2}.
\end{eqnarray*}
The second inequality (\ref{eqn:unweighted-2-weighted_restricted_upper_bound}) follows from assumption 3(b) and Lemma \ref{lem:diag_weighted_spectral_bound} applying to $\hat\Sigma_t = |N_t| {{\cal X}_t^\diamond}^\top W_t {{\cal X}_t^\diamond}$ and $W_t \ge 0$.
\end{proof}

\begin{lem}
\label{lem:tv-lasso-rate}
Suppose assumption 1, 2, 3 and 5(a) hold. Let $t \in \varpi$ be fixed and $\lambda_0$ be defined in (\ref{eqn:lambda_0+L_t}).
Then, for $\lambda_1 \ge 2(\lambda_0 + 2 C_0 L_{t,1} s^{1/2} \varepsilon_0^{-1} b_n |N_t| \overline{w}_t)$ where $\lambda_0$ is defined in (\ref{eqn:lambda_0+L_t}), we have, with probability $1-2p^{-1}$,
\begin{equation}
\label{eqn:tv-lasso-rate}
|{\cal X}_t [\tilde{\mbf\beta}(t) - {\mbf\beta}(t)]|_2^2 + \lambda_1  |\tilde{\mbf\beta}(t) - {\mbf\beta}(t)|_1 \le 4 \lambda_1^2 {s \over \phi_0^2}.
\end{equation}
\end{lem}

\begin{proof}
By definition (\ref{eqn:tv-lasso}),
\begin{equation*}
|{\cal Y}_t - {\cal X}_t \tilde{\mbf\beta}(t)|_2^2 + \lambda_1 |\tilde{\mbf\beta}(t)|_1 \le |{\cal Y}_t - {\cal X}_t {\mbf\beta}(t)|_2^2 + \lambda_1 |{\mbf\beta}(t)|_1,
\end{equation*}
which implies that
\begin{equation*}
|{\cal X}_t [\tilde{\mbf\beta}(t) - {\mbf\beta}(t)]|_2^2 + \lambda_1 |\tilde{\mbf\beta}(t)|_1 \le \lambda_1 |{\mbf\beta}(t)|_1 + 2 \left\langle {\cal Y}_t - {\cal X}_t {\mbf\beta}(t), {\cal X}_t [\tilde{\mbf\beta}(t) - {\mbf\beta}(t)] \right \rangle.
\end{equation*}
By assumption 2 and Taylor's expansion in the $b_n$-neighborhood of $t$, we see that 
\begin{equation}
\label{eqn:taylor}
{\cal Y}_t - {\cal X}_t {\mbf\beta}(t) = {\cal E}_t + M_t {\cal X}_t \mbf\beta'(t) + {\cal X}_t \mbf\xi,
\end{equation}
where $M_t = \diag((t_i-t)_{i \in N_t})$ and $\mbf\xi$ is a vector such that $|\mbf\xi|_\infty \le C_0 b_n^2 / 2$ and $|\mbf\xi|_0 \le s$. Let ${\cal J} = \{2 |{\cal E}_t^\top {\cal X}_t|_\infty  \le \lambda_0\}$. Observe that $|{\cal E}_t^\top {\cal X}_t|_\infty = \max_{j \le p} |\sum_{i \in N_t} w(i,t) X_{ij} e_i|$ and, by assumption 1,
\begin{equation}
\label{eqn:concentration_iid_gaussian}
\sum_{i \in N_t} w(i,t) X_{ij} e_i \sim N\left( 0, \sigma^2 \sum_{i \in N_t} w(i,t)^2 X_{ij}^2 \right).
\end{equation}
Then, by the standard Gaussian tail bound and the union bound, we obtain that
\begin{equation*}
\Prob \left( \max_{j \le p} \left| {\sum_{i \in N_t} w(i,t) X_{ij} e_i \over \sigma L_{t,2}} \right| \ge \sqrt{\varepsilon^2 + 2 \log{p}} \right) \le \Prob(\max_{j \le p} |Z_j| \ge \sqrt{\varepsilon^2 + 2 \log{p}}) \le 2 \exp\left(-{\varepsilon^2 \over 2}\right)
\end{equation*}
for all $\varepsilon > 0$, where $Z_j \sim N(0,1)$. Now, choose $\varepsilon = (2\log{p})^{1/2}$ and $\lambda_0 = 4 \sigma L_{t,2} (\log{p})^{1/2}$, we have $\Prob({\cal J}) \ge 1 - 2 p^{-1}$. Further, we have
\begin{eqnarray*}
&& |\mbf\beta'(t)^\top {\cal X}_t^\top M_t {\cal X}_t [\tilde{\mbf\beta}(t)-\mbf\beta(t)]| \le |\tilde{\mbf\beta}(t)-\mbf\beta(t)|_1 |{\cal X}_t^\top M_t {\cal X}_t \mbf\beta'(t)|_\infty \\
&\le& |\tilde{\mbf\beta}(t)-\mbf\beta(t)|_1 \max_{j \le p} \left( \sum_{i \in N_t} w(i,t) X_{ij}^2 \right)^{1/2} \left[ \mbf\beta'(t)^\top {\cal X}_t^\top M_t^2 {\cal X}_t \mbf\beta'(t) \right]^{1/2} \quad (\text{Cauchy-Schwarz}) \\
&\le& |\tilde{\mbf\beta}(t)-\mbf\beta(t)|_1 L_{t,1} \sqrt{\rho_{\max}({\cal X}_t^\top M_t^2 {\cal X}_t, s)} |\mbf\beta'(t)|_2 \quad (\text{assumption 2}) \\
&\le& |\tilde{\mbf\beta}(t)-\mbf\beta(t)|_1 L_{t,1} C_0 s^{1/2} b_n \sqrt{\rho_{\max}({\cal X}_t^\top {\cal X}_t,s )} \quad (\text{Lemma \ref{lem:diag_weighted_spectral_bound}, assumption 2 and 3}) \\
&\le& |\tilde{\mbf\beta}(t)-\mbf\beta(t)|_1 L_{t,1} C_0 (|N_t| \overline{w}_t s)^{1/2} b_n \varepsilon_0^{-1}  \quad (\text{Lemma \ref{lem:unweighted-2-weighted}, equation (\ref{eqn:unweighted-2-weighted_restricted_upper_bound})}).
\end{eqnarray*}
Similarly, we can show that $|\mbf\xi^\top {\cal X}_t^\top {\cal X}_t [\tilde{\mbf\beta}(t)-\mbf\beta(t)]| = O( L_{t,1} (|N_t| \overline{w}_t s)^{1/2} b_n^2 \varepsilon_0^{-1} |\tilde{\mbf\beta}(t)-\mbf\beta(t)|_1)$. Therefore, it follows that, with probability at least $(1-2p^{-1})$,
\begin{align*}
\left| \left\langle {\cal Y}_t - {\cal X}_t {\mbf\beta}(t), {\cal X}_t [\tilde{\mbf\beta}(t) - {\mbf\beta}(t)] \right \rangle \right| \le \left[ \lambda_0 + 2 L_{t,1} C_0 (|N_t| \overline{w}_t s)^{1/2} b_n \varepsilon_0^{-1} (1+o(1)) \right] |\tilde{\mbf\beta}(t)-\mbf\beta(t)|_1.
\end{align*}
Now, choose $\lambda_1 \ge 2(\lambda_0 + 2 L_{t,1} C_0 (|N_t| \overline{w}_t s)^{1/2} b_n \varepsilon_0^{-1})$, we get
\begin{equation*}
2|{\cal X}_t [\tilde{\mbf\beta}(t) - {\mbf\beta}(t)]|_2^2 + 2\lambda_1 |\tilde{\mbf\beta}(t)|_1 \le \lambda_1 |\tilde{\mbf\beta}(t)-\mbf\beta(t)|_1 + 2 \lambda_1 |\mbf\beta(t)|_1.
\end{equation*}
Denote $S_0 := S_0(t) = \text{supp}({\mbf\beta}(t))$. By the same argument as \cite[Lemma 6.3]{buhlmannvandegeer2011}, it is easy to see that, on ${\cal J}$,
\begin{equation*}
2|{\cal X}_t [\tilde{\mbf\beta}(t) - {\mbf\beta}(t)]|_2^2 + \lambda_1 |\tilde{\mbf\beta}_{S_0^c}(t)|_1 \le 3 \lambda_1 |\tilde{\mbf\beta}_{S_0}(t) - {\mbf\beta}_{S_0}(t)|_1.
\end{equation*}
But then, (\ref{eqn:tv-lasso-rate}) follows from the restricted eigenvalue condition (assumption 5) with the elementary inequality $4ab \le a^2 + 4b^2$ that
\begin{equation*}
2|{\cal X}_t [\tilde{\mbf\beta}(t) - {\mbf\beta}(t)]|_2^2 + \lambda_1  |\tilde{\mbf\beta}(t) - {\mbf\beta}(t)|_1 \le  4 \lambda_1 |\tilde{\mbf\beta}_{S_0}(t) - {\mbf\beta}_{S_0}(t)|_1 \le |{\cal X}_t [\tilde{\mbf\beta}(t) - {\mbf\beta}(t)]|_2^2 + 4 \lambda_1^2 {s / \phi_0^2}.
\end{equation*}
\end{proof}

\begin{defn}
\label{defn:subgaussian_rv}
A mean zero random variable is said to be {\it sub-Gaussian} with variance factor $\sigma^2$ if
\begin{equation*}
\log \E (e^{\lambda X}) \le \lambda^2 \sigma^2 / 2 \qquad \text{for all } \lambda \in \mathbb{R}.
\end{equation*}
\end{defn}

\begin{lem}
\label{lem:concentration_GP}
Let $\xi_{i}$ be iid sub-Gaussian random variables with mean zero and variance factor $\sigma^2$, and $e_i = \sum_{m=0}^\infty a_m \xi_{i-m}$ be a linear process. Let $\vw = (w_1, \cdots, w_n)$ be a real vector and $S_n = \sum_{i=1}^n w_i e_i$ be the weighted partial sum of $e_i$.
\begin{enumerate}
\item (Short-range dependence). If $|\va|_1 = \sum_{i=0}^\infty |a_i| < \infty$, then for all $x > 0$ we have
\begin{equation}
\label{eqn:concentration_GP_SRD}
\Prob(|S_n| \ge x) \le 2 \exp\left( - {x^2 \over 2 |\vw|_2^2 |\va|_1^2 \sigma^2} \right).
\end{equation}

\item (Long-range dependence). Suppose $K = \sup_{m \ge 0} |a_m| (m+1)^\varrho < \infty$, where $1/2 < \varrho < 1$. Then, there exists a constant $C_\varrho$ that only depends on $\varrho$ such that
\begin{equation}
\label{eqn:concentration_GP_LRD}
\Prob(|S_n| \ge x) \le 2 \exp\left( - {C_\varrho x^2 \over |\vw|_2^2 n^{2(1-\varrho)} \sigma^2 K^2} \right).
\end{equation}
\end{enumerate}
\end{lem}

\begin{proof}
Put $a_m = 0$ if $m < 0$ and we may write $S_n = \sum_{m \in \mathbb{Z}} b_m \xi_m$, where $b_m = \sum_{i=1}^n w_i a_{i-m}$. By the Cauchy-Schwarz inequality,
\begin{equation*}
\sum_{m \in \mathbb{Z}} b_m^2 \le \sum_{m \in \mathbb{Z}} \left(\sum_{i=1}^n w_i^2 |a_{i-m}| \right) \left(\sum_{i=1}^n |a_{i-m}|\right) \le |\vw|_2^2 |\va|_1^2.
\end{equation*}
Then, (\ref{eqn:concentration_GP_SRD}) follows from the Cram\'er-Chernoff bound \cite{blm2013}. Let $\bar{a}_m = \max_{l \ge m} |a_l|$ and $A_m = \sum_{l=0}^m |a_l|$. Note that $A_n \le K \sum_{l=0}^n (l+1)^{-\varrho} \le C_\varrho K (n+1)^{1-\varrho}$, where $C_\varrho = (1-\varrho)^{-1}$. Then, we have
\begin{equation*}
\sum_{m=1-n}^n b_m^2 \le \sum_{m=1-n}^n \left(\sum_{i=1}^n w_i^2 |a_{i-m}| \right) \left(\sum_{i=1}^n |a_{i-m}|\right) \le |\vw|_2^2 A_{2n}^2.
\end{equation*}
If $m \le -n$, then $|b_m|\le |\vw|_1 \bar{a}_{1-m}$ and therefore
\begin{equation*}
\sum_{m \le -n} b_m^2 \le |\vw|_1^2 \sum_{m \le -n} \bar{a}_{1-m}^2 \le C_\varrho n |\vw|_2^2 K^2 n^{1-2\varrho},
\end{equation*}
where the last inequality follows from Karamata's theorem; see e.g. \cite{resnick1987a}. Hence, the proof is complete by invoking the Cram\'er-Chernoff bound for sub-Gaussian random variables.
\end{proof}

\subsection*{Some implementation issues}
\label{sec:extensions}

We assumed that the noise variance-covariance matrix $\Sigma_e$ is known. In the iid error case $\Sigma_e=\sigma^2 I_n$, we have seen that the distribution $F(\cdot)$ is independent of $\sigma^2$ and therefore its value does not affect the inference procedure. The noise variance only impacts the tuning parameter of the initial Lasso estimator. In practice, we can use the scaled Lasso to estimate $\sigma^2$ in our numeric studies. Given that $|\hat\sigma / \sigma - 1| = o_\Prob(1)$ \cite{sunzhang2012}, the theoretical properties of our estimator (\ref{eqn:proposed-estimator}) remains the same if we plug in the scaled Lasso variance output to our method. %In our simulation studies, the tuning parameter $\lambda_1$ for the scaled Lasso is fixed to $K (\log{p} / (n b_n))^{1/2}$ for some numeric constant $K$.
For temporally dependent stationary error process, estimation of $\Sigma_e$ becomes more subtle since it involves $n$ autocovariance parameters. We propose a heuristic strategy: first, run the tv-Lasso estimator and get the residuals; then calculate the sample autocovariance matrix and apply a banding or tapering operation $B_h(\Sigma) = \{\sigma_{jk} \vone(|j-k|\le h) \}_{j,k=1}^p$ \cite{bickellevina2008b,caizhangzhou2010a,mcmurrypolitis2010}.

We provide some justification on the heuristic strategy for SRD time series models. To simplify explanation, we consider the uniform kernel and the bandwidth $b_n=1$. Suppose we have an oracle where $\mbf\beta(t)$ is known and we have access to the error process $e(t)$. Let $\Sigma^*_e$ be the oracle sample covariance matrix of $e_i$ with the Toeplitz structure i.e. the $h$-th subdiagonal of $\Sigma^*_e$ is $\sigma^*_{e,h} = n^{-1} \sum_{i=1}^{n-h} e_i e_{i+h}$. We first compare the oracle estimator and the true error covariance matrix $\Sigma_e$. Let $\alpha>0$ and define
\begin{equation*}
{\cal T}(\alpha, C_1, C_2) = \left\{ M \in ST^{p\times p} : \sum_{k=h+1}^p |m_k| \le C_1 h^{-\alpha}, \rho_j(M) \in [C_2, C_2^{-1}],\; \forall j=1,\cdots,p \right\},
\end{equation*}
where $ST^{p\times p}$ is the set of all $p \times p$ symmetric Toeplitz matrices. If $e_i$ has SRD, then $\Sigma_e \in {\cal T}(\varrho-1,C_1,C_2)$. By the argument in \cite{bickellevina2008b} and Lemma \ref{lem:concentration_GP}, we can show that
\begin{eqnarray*}
\rho_{\max}(B_h(\Sigma^*_e)-\Sigma_e) &\le& \rho_{\max}(B_h(\Sigma^*_e)-B_h(\Sigma_e)) + \rho_{\max}(B_h(\Sigma_e)-\Sigma_e) \\
&\lesssim_\Prob& h \sqrt{\log{h} \over n} + h^{-(\varrho-1)}.
\end{eqnarray*}
Choosing $h^*\asymp(n/\log{n})^{1/(2\varrho)}$, we get
$$
\rho_{\max}(B_h(\Sigma^*_e)-\Sigma_e) = O_\Prob\left( \left({\log{n}\over n} \right)^{\varrho-1\over 2\varrho} \right).
$$
This oracle rate is sharper than the one established in \cite{bickellevina2008b} for regularizing more general bandable matrices if $n=o(p)$. Here, the improved rate is due to the Toeplitz structure in $\Sigma_e$. Since $\Sigma_e$ has uniformly bounded eigenvalues from zero and infinity, the banded oracle estimator $B_h(\Sigma^*_e)$ can be used as a benchmark to assess the tv-Lasso residuals $\tilde{\cal E}_t = {\cal Y}_t - {\cal X}_t \tilde{\mbf\beta}(t)$.

\begin{prop}
\label{prop:error_cov_mat_bandwidth_selection}
Suppose $\Sigma_e \in {\cal T}(\varrho-1, C_1, C_2)$ and conditions of Lemma \ref{lem:tv-lasso-rate} are satisfied except that $(e_i)$ is an SRD stationary Gaussian process with $\varrho>1$. Then
\begin{equation}
\label{eqn:error_cov_mat_bandwidth_selection}
\rho_{\max}(B_h(\hat\Sigma_e) - B_h(\Sigma^*_e)) = O_\Prob(h \lambda_1 s^{1/2}).
\end{equation}
With the choice $h^*\asymp(n'/\log{n'})^{1/2\varrho}$ where $n'=|N_t|$, we have
\begin{equation}
\label{eqn:error_cov_mat_bandwidth_selection_2}
\rho_{\max}(B_h(\hat\Sigma_e) -\Sigma_e)) = O_\Prob \left( \left({\log{n'}\over n'} \right)^{\varrho-1\over 2\varrho} + \left({n' \over \log{n'}} \right)^{1\over 2\varrho} \left(\sqrt{s\log{p} \over n'} + s b_n)\right) \right).
\end{equation}
\end{prop}
It is interesting to note that the price we pay to choose $h$ for not knowing the error process is the second term in (\ref{eqn:error_cov_mat_bandwidth_selection_2}). Bandwidth selection for the smoothing parameter $b_n$ is a theoretically challenging task in the high dimension. Asymptotic optimal order for the parameter is available up to some unknown constants depending on the data generation parameters. We shall use the cross-validation (CV) in our simulation studies and real data analysis.

\begin{proof}[Proof of Proposition \ref{prop:error_cov_mat_bandwidth_selection}]
Since we consider the uniform kernel, we may assume $b_n=1, |N_t|=n$ and then rescale. Observe that
\begin{eqnarray*}
\max_{|k| \le h} |\hat\sigma_{e,k}^2-\sigma_{e,k}^{*2}| &=& \max_{|k| \le h} {1\over n} \left|\sum_{i=1}^{n-k} (\hat{e}_i\hat{e}_{i+k}-e_ie_{i+k})\right| \\
&\le& \max_{|k| \le h} {1\over n} \left|\sum_{i=1}^{n-k} \hat{e}_i(\hat{e}_{i+k}-e_{i+k})\right| + \left|\sum_{i=1}^{n-k} e_{i+k}(\hat{e}_i-e_i)\right| \\
&\le& \max_{|k| \le h} {1\over n} \left(\sum_{i=1}^{n-k} \hat{e}_i^2 \right)^{1/2}  \left(\sum_{i=1}^{n-k} (\hat{e}_{i+k}-e_{i+k})^2 \right)^{1/2} \\
&& \quad + \max_{|k| \le h} {1\over n} \left(\sum_{i=1}^{n-k} e_{i+k}^2 \right)^{1/2}  \left(\sum_{i=1}^{n-k} (\hat{e}_i-e_i)^2 \right)^{1/2} \\
&\le& \left[ \left({1\over n} \sum_{i=1}^n \hat{e}_i^2 \right)^{1/2} + \left( {1\over n} \sum_{i=1}^n e_i^2 \right)^{1/2} \right]  \left({1\over n} \sum_{i=1}^n (\hat{e}_i-e_i)^2 \right)^{1/2}.
\end{eqnarray*}
By Lemma \ref{lem:tv-lasso-rate},
$$
{1\over n} \sum_{i=1}^n (\hat{e}_i-e_i)^2 = |\tilde{\cal E}_t - {\cal E}_t |_2^2 = |\calX_t[\tilde{\mbf\beta}(t)-\mbf\beta(t)]|_2^2 = O_\Prob(\lambda_1^2 s).
$$
%Then, we have
%$$
%|\hat\sigma_{e,0}^2-\sigma_{e,0}^{*2}| = {1\over n} \left|\sum_{i=1}^n (\hat{e}_i^2-e_i^2) \right| \le \left( {1\over n} \sum_{i=1}^n(\hat{e}_i-e_i)^2 \right)^{1/2} \left( {1\over n} \sum_{i=1}^n(\hat{e}_i+e_i)^2 \right)^{1/2} = O_\Prob(\lambda_1 s^{1/2}),
%$$
Then, it follows from the last expression and $n^{-1} \sum_{i=1}^n e_i^2 = O_\Prob(1)$ that
$$
\max_{|k| \le h} |\hat\sigma_{e,k}^2-\sigma_{e,k}^{*2}| = O_\Prob(\lambda_1 s^{1/2}).
$$
Therefore
$$
\rho_{\max}(B_h(\hat\Sigma_e) - B_h(\Sigma^*_e)) \lesssim h \max_{|k| \le h} |\hat\sigma_{e,k}^2-\sigma_{e,k}^{*2}| = O_\Prob(h \lambda_1 s^{1/2}).
$$
\end{proof}

\end{document}